\documentclass[11pt]{article}

\usepackage[normalem]{ulem}

\usepackage[usenames,dvipsnames,svgnames,table]{xcolor}
\usepackage[english]{babel}
\usepackage[hidelinks]{hyperref}

\usepackage{amsmath,amssymb,amsthm}
\usepackage{mathtools}
\usepackage{graphicx}
\usepackage{enumitem}
\usepackage{braket}

\usepackage[longend,ruled,linesnumbered,nokwfunc]{algorithm2e}
\SetKwInput{Input}{Input}
\SetKwInput{Parameters}{Parameters}
\SetKwInput{Output}{Output}
\SetKwInput{Guarantee}{Guarantee}
\SetKwInput{Analysis}{Analysis}
\SetKw{True}{True}
\SetKw{False}{False}
\usepackage[capitalize]{cleveref}
\crefalias{AlgoLine}{line}
\crefname{line}{line}{lines}

\usepackage{palatino}
\usepackage{mathpazo}
\usepackage{fullpage}
\hypersetup{pdfpagemode=UseNone}
\usepackage[margin=1cm]{caption}
\usepackage[parfill]{parskip} 

\usepackage{thmtools}
\usepackage{thm-restate}

\newcommand{\ignore}[1]{}

\renewcommand{\S}{\mathcal{S}}

\newcommand{\R}{\mathbb{R}}

\newcommand{\E}{\mathbb{E}}

\DeclareMathOperator{\Tr}{Tr}
\DeclareMathOperator{\nnz}{nnz}

\renewcommand{\sup}{\mathrm{sup}}
\renewcommand{\min}{\mathrm{min}}

\renewcommand{\max}{\mathrm{max}}
\renewcommand{\epsilon}{\varepsilon}
\newcommand{\eps}{\epsilon}

\def\01{\{0,1\}}

\newtheorem{defin}{Definition}[section]
\newtheorem{definition}[defin]{Definition}

\newtheorem*{proposition*}{Proposition}
\newtheorem{theorem}[defin]{Theorem}
\newtheorem*{theorem*}{Theorem}
\newtheorem{remark}[defin]{Remark}

\newtheorem{lemma}[defin]{Lemma}

\newtheorem*{claim*}{Claim}

\newtheorem*{conjecture*}{Conjecture}
\newtheorem*{theoremintro}{Theorem 3.1}
\theoremstyle{definition}

\newcommand{\wt}{\widetilde}
\DeclareMathOperator{\Diag}{Diag}
\DeclareMathOperator{\diag}{Diag}

\DeclareMathOperator{\poly}{poly}
\DeclareMathOperator{\polylog}{polylog}
\newcommand{\tO}{\widetilde{O}}

\DeclareMathOperator{\MAJ}{MAJ}
\DeclareMathOperator{\OR}{OR}

\begin{document}

\title{Quantum speedups for linear programming \\ via interior point methods}
\date{}
\author{}
\author{Simon Apers\thanks{Universit\'e Paris Cit\'e, CNRS, IRIF, Paris, France. \texttt{apers@irif.fr}} \and Sander Gribling\thanks{Tilburg University, Tilburg, the Netherlands. \texttt{s.j.gribling@tilburguniversity.edu}}}

\maketitle

\begin{abstract}
We describe a quantum algorithm based on an interior point method for solving a linear program with $n$ inequality constraints on $d$ variables.
The algorithm explicitly returns a feasible solution that is $\varepsilon$-close to optimal, and runs in time $\sqrt{n} \cdot \poly(d,\log(n),\log(1/\varepsilon))$ which is sublinear for tall linear programs (i.e., $n \gg d$).
Our algorithm speeds up the Newton step in the state-of-the-art interior point method of Lee and Sidford [FOCS~'14].
This requires us to efficiently approximate the Hessian and gradient of the barrier function, and these are our main contributions.

To approximate the Hessian, we describe a quantum algorithm for the \emph{spectral approximation} of $A^T A$ for a tall matrix $A \in \mathbb R^{n \times d}$.
The algorithm uses leverage score sampling in combination with Grover search, and returns a $\delta$-approximation by making $O(\sqrt{nd}/\delta)$ row queries to $A$.
This generalizes an earlier quantum speedup for graph sparsification by Apers and de Wolf~[FOCS~'20].
To approximate the gradient, we use a recent quantum algorithm for multivariate mean estimation by Cornelissen, Hamoudi and Jerbi [STOC '22].
While a naive implementation introduces a dependence on the condition number of the Hessian, we avoid this by pre-conditioning our random variable using our quantum algorithm for spectral approximation.
\end{abstract}

\clearpage 

\tableofcontents

\clearpage

\section{Introduction} 
 
We consider the fundamental task of solving linear programs (LPs) of the form 
\[
\min_{x}\  c^T x \quad \text{s.t.} \quad Ax \geq b
\]
where $A \in \R^{n \times d}$, $b \in \R^n$, $c \in \R^d$, and the minimization is over $x \in \R^d$. Famously, Khachiyan~\cite{Khachiyan79} showed how to use the ellipsoid method to solve such linear programs efficiently, i.e., in polynomial time.
The result however was mostly of theoretical relevance, because the ellipsoid method runs very slow in practice.
It was a second breakthrough result by Karmarkar \cite{Karmarkar84} that introduced \emph{interior point methods} (IPMs) as a provably polynomial time algorithm for solving LPs that is also efficient in practice.
To this day the study of IPMs is central to our understanding of fast algorithms for solving LPs, see e.g.~\cite{Renegar01,NesterovNemirovskii93,LeeSidford19,brand2020solving,brand2021minimum,cohen2021solving,shunhua21fasterLP}.
In the regime that we focus on, with many constraints and few variables, the classical state-of-the-art asymptotic complexity for solving an LP is achieved by an IPM of van den Brand, Lee, Liu, Saranurak, Sidford, Song and Wang~\cite{brand2021minimum}, which runs in time $\tO(nd + d^{2.5})$. When the number of constraints is comparable to the number of variables, the IPM of Cohen, Lee and Song achieves the state-of-the-art complexity~$\tO(\max\{n,d\}^{\max\{\omega,\ 2.5-\alpha/2,\ 2+1/6\}})$ where $\omega$ and $\alpha$ are the exponent of matrix multiplication and its dual~\cite{cohen2021solving}.\footnote{The exponent $2+1/6$ has subsequently been improved to $2+1/18$~\cite{shunhua21fasterLP}.}

With the growing interest in using quantum computers to solve optimization tasks, it is natural to ask whether quantum algorithms can speed up IPMs. In a nutshell, we give a positive answer for IPMs that solve ``tall'' LPs. We describe a quantum IPM that solves (full-dimensional) linear programs with $d$ variables and $n$ inequality constraints in time $\sqrt{n} \cdot \poly(d,\log(1/\eps))$, given quantum query access to the problem data. This is sublinear in the input size when the LP is tall ($n$ much larger than $d$). To obtain our quantum speedup we develop two novel quantum subroutines: 
\begin{enumerate}
    \item[(i)] an algorithm to obtain a spectral approximation of $B^T B$ for tall $B$ $\in \R^{n \times d}$ (i.e.~$n \gg d$), and
    \item[(ii)] an approximate matrix-vector product algorithm.
\end{enumerate}
While these new routines have their own independent interest (e.g., for computing statistical leverage scores), we use them to approximate the costly Newton steps at the core of an IPM.
More specifically, we use these algorithms to approximate Hessians and gradients of three canonical \emph{self-concordant barriers}. This enables us to perform an approximate Newton step, which is at the core of IPMs. 

The remainder of the introduction is organized as follows. We first give a brief overview of the type of IPMs we use and the landmark results in that area, focusing only on the theoretical developments and not the practical contributions. We then describe our quantum algorithms in more detail. This leads us to describe our main results and put them in context. 
We finally discuss related work on quantum algorithms for optimization and linear algebra.

\subsection{Interior point methods for linear programs} 

In a seminal work, Nesterov and Nemirovskii linked interior point methods to \emph{self-concordant barriers}~\cite{NesterovNemirovskii93}. We defer their definition until later -- for this introduction one can think of them as convex functions that go to infinity as we approach the boundary of the feasible region; a standard example is the logarithmic barrier $f(x) = - \sum_{i=1}^n \log(a_i^T x- b_i)$ where $a_i^T$ is the $i$th row of $A$. Given a barrier~$f$, one can define the family of unconstrained problems parameterized by $\eta>0$
\[
\min_x\  f_\eta(x) \quad \text{where } f_\eta(x) = \eta c^T x + f(x).
\] 
For small $\eta$, the function effectively only takes into account the constraints (i.e., the barrier function) and the minimizer will be close to some geometric ``center'' of the feasible region, while for $\eta$ sufficiently large the minimizer approaches that of the LP. Optimizing $f_\eta$ for a fixed value of $\eta$ is typically done using \emph{Newton's method}.
Newton's method approaches the optimum by making updates of the form $x' = x + n(x)$, with Newton step
\begin{equation} \label{eq:Newton-step-intro}
n(x)
= - H(x)^{-1} g(x),
\end{equation}
where $H(x) = \nabla^2 f(x)$ is the Hessian at $x$, and $g(x) = \nabla f(x)$ is the gradient at $x$. Self-concordant functions $f$ have the following desirable property: if we are close to a minimizer for $f_\eta$, then we can multiplicatively increase $\eta$ by a small amount to $\eta'$ and still quickly converge to a  minimizer of the new $f_{\eta'}$ using (approximate) Newton steps. The complexity of such an IPM is thus determined by the rate at which we can increase $\eta$ and the cost of an (approximate) Newton step. 

A self-concordant barrier $f$ has an associated complexity parameter $\vartheta_f$ that governs this rate: we can increase $\eta$ multiplicatively with a factor roughly $1+1/\sqrt{\vartheta_f}$. Starting from a point close to the minimizer of $f_1$, the number of times we need to increase $\eta$ scales roughly as $\tO(\sqrt{\vartheta_f} \log(1/\eps))$.\footnote{Throughout we use $\tO$-notation to hide polylogarithmic factors in $n$, $d$, and $1/\eps$. For simplicity we count unit time for a query to an entry of $A$, $b$, or $c$.}

The standard \emph{logarithmic barrier} used in~\cite{Renegar88} has a complexity of~$n$, the number of inequality constraints. In this work, we focus on the regime where $n \gg d$. In that setting, Vaidya showed how to obtain an improved complexity using the \emph{volumetric barrier}~\cite{Vaidya96} -- a linear combination of that barrier with the logarithmic barrier has an improved complexity of $O(\sqrt{nd})$~\cite{VaidyaAtkinson93}. In their seminal work, Nesterov and Nemirovskii showed that each closed convex domain in $\R^d$ admits a \emph{universal barrier} whose complexity is $O(d)$~\cite{NesterovNemirovskii93}. The iterations of an IPM based on the universal barrier are unfortunately not efficiently computable. Roughly 20 years later, Lee and Sidford defined a self-concordant barrier based on Lewis weights, the \emph{Lewis weight barrier}, which has complexity $O(d \polylog(n))$~\cite{LeeSidford19}; based on this barrier they designed an IPM that uses $O(\sqrt{d} \log(1/\eps))$ iterations, each of which involves solving $\tO(1)$ linear systems. We summarize the three landmark barriers and their Hessians and gradients in \cref{tab:self-concordant barriers}. We define and discuss leverage scores and Lewis weights in more detail in \cref{sec:leverage scores and lewis weights}.

\begin{table}[ht!]
    \centering
    \begin{tabular}{l|l|l|l}
         Barrier &  Complexity $\vartheta$ & Hessian $H(x)$ & Gradient $g(x)$ \\ \hline 
       Logarithmic barrier  & $n$ & $A^T S_x^{-2} A$ & $-A^T S_x^{-1} \mathbf{1}$ \\
       Volumetric / hybrid barrier & $\sqrt{nd}$ & $\approx A^T S_x^{-1} \Sigma_x S_x^{-1} A$ & $-A^T S_x^{-1} \Sigma_x \mathbf{1}$ \\
       $\ell_p$-Lewis weight barrier & $d \cdot \polylog(n)$ & $\approx A^T S_x^{-1} W_x S_x^{-1} A$ & $-A^T S_x^{-1} W_x \mathbf{1}$
    \end{tabular}
    \caption{Self-concordant barriers, their complexity as a function of the number of variables $d$ and constraints $n$, and expressions for the Hessian and gradient. Diagonal matrices $S_x$, $\Sigma_x$, and $W_x$ contain in the $i$th diagonal entry respectively the $i$th slack $(Ax-b)_i$, the leverage score $\sigma_i(S_x^{-1} A)$, or the $\ell_p$-Lewis weight $w^{(p)}_i(S_x^{-1} A)$. $A \approx B$ is shorthand for $A \preceq B \preceq \widetilde O(1) A$.}
    \label{tab:self-concordant barriers}
\end{table}

It is well known that IPMs do not require an exact implementation of the Newton step. Even the earliest analyses imply they can tolerate approximations of the Newton step in the local norm up to small constant error (say $0.01$). Here we use one additional type of approximation: we first spectrally approximate the Hessian. In \cref{sec:robustness} we give a formal statement that shows it suffices to implement approximations of the Newton steps, based on approximations of both the Hessians and gradients in \cref{tab:self-concordant barriers}; the argument is relatively standard. Similar arguments have recently lead to a space-efficient IPM in \cite{Space2023}. 
Specifically, we will implement approximate Newton steps of the form $\tilde n(x) = - Q(x)^{-1} \tilde g(x)$, where
\begin{equation} \label{eq:Newtons}
Q(x) \approx H(x)
\quad \text{ and } \quad
\| \tilde g(x) - g(x) \|_{H(x)^{-1}}
\leq \delta,
\end{equation}
for $1/\delta \in \tO(1)$, where we use $\| v \|_A$ to denote $\sqrt{v^T A v}$ and $A \approx B$ for $A \preceq B \preceq \widetilde O(1) A$. 

\subsection{Quantum algorithms}

Here we summarize the different quantum algorithms that we combine to approximate Newton steps in the sense of \cref{eq:Newtons}.

\subsubsection{Quantum spectral approximation}

As mentioned before, our first task is to spectrally approximate a matrix $B^T B$ where $B \in \R^{n \times d}$ and $n \gg d$. 
We prove the following theorem, where a row query to $B$ corresponds to learning all the entries of a row of~$B$.  

\begin{theoremintro}[Quantum spectral approximation] 
Consider query access to a matrix $B \in \R^{n \times d}$ with row sparsity $r$.
For any $0 < \eps \leq 1$, there is a quantum algorithm that, with high probability, returns a matrix $\widetilde B \in \R^{\tO(d/\eps^2) \times d}$ satisfying
\[
(1-\eps) \widetilde B^T  \widetilde B
\preceq B^T B
\preceq (1+\eps) \widetilde B^T \widetilde B,
\]
while making $\tO(\sqrt{nd}/\eps)$ row queries to $B$, and taking time $\tO(r \sqrt{nd}/\eps + d^\omega)$.
\end{theoremintro}

For comparison, the classical state of the art is $\widetilde O(\mathrm{nnz}(B) + d^3)$ where $\mathrm{nnz}(B)$ is the number of non-zero entries of $B$~\cite{clarkson2017subspace}. We thus achieve a quantum speedup for dense instances that have $n \gg d$. It is not hard to see that any classical algorithm for spectral approximation must make $\Omega(n)$ row queries to~$B$.\footnote{E.g., consider $B \in \R^{n \times 1}$ with $B$ either a vector with all zeros, or a vector with all zeros but a single 1. If we think of~$B$ as an $n$-bit string, then a spectral approximation of $B^T B$ computes the OR function on $B$, which is known to require $\Omega(n)$ classical queries, but only $O(\sqrt{n})$ quantum queries via Grover's search algorithm~\cite{grover1996QSearch}. As we show in \cref{sec: LB}, an $\Omega(n)$ query lower bound holds in a quantum-inspired classical access model as well.} In a nutshell, our $n \rightarrow \sqrt{nd}$ quantum speedup comes from using Grover search to find the $\tO(d)$ important rows among the $n$ rows of $B$ in $\tO(\sqrt{nd})$ row queries to $B$. 

In this work, we will use this in the context of IPMs to approximate the Hessian at each iterate. As can be seen in~\cref{tab:self-concordant barriers}, the Hessians that we consider take the form $H = B^T B$, with $B \in \R^{n \times d}$ equal to some diagonal matrix of weights times the constraint matrix $A$.

More broadly, such spectral approximations are central to dimensionality reduction techniques in several areas. For example, in data science, where $B$ represents a tall data matrix, a spectral approximation of $B^T B$ can be used for statistical regression \cite{YinTatThesis,parulekar2021L1}. In graph theory, spectral approximations can be used to sparsify a graph. We note that the above theorem generalizes an earlier quantum algorithm for graph sparsification \cite{apers2022quantum}, although it uses very different techniques.
Indeed, we can recover the complexity in \cite[Theorem~1]{apers2022quantum} by letting $A \in \R^{n \times d}$ denote the edge-vertex incidence matrix of a graph with $n$ edges and $d$ vertices, which has sparsity $r = 2$. For graph sparsification, a matching lower bound is known~\cite[Theorem~2]{apers2022quantum}, see \cref{remark:graph-sparsification} for more details.

At a high level, our algorithm does row sampling based on statistical {\it leverage scores}. The leverage score of the $i$-th row $b_i^T$ of $B$ is defined as
\[
\sigma_i(B)
= b_i^T (B^T B)^+ b_i,
\]
and this measures the ``statistical importance'' of a row with respect to the full set of rows~\cite{Mahoney11,li2013iterative}. Here $C^+$ denotes the Moore-Penrose pseudoinverse of $C$.
A matrix $\widetilde B$ is then constructed by sampling and rescaling a subset of $\tO(d)$ rows of~$B$, in such a way that ensures $\widetilde B^T \widetilde B \approx B^T B$. Typically one includes each row $i$ independently with a probability that is roughly proportional to $\sigma_i(B)$. 
Naively computing these scores $\sigma_i(B)$ however already requires knowledge of $B^T B$. 
We bypass this issue by building on a recent, elegant bootstrapping technique by Cohen, Lee, Musco, Musco, Peng and Sidford~\cite{cohen2015uniform}.
Implementing their technique in the quantum setting, in sublinear time, is challenging and one of our main contributions.

As an example application, we can use this theorem to get a spectral approximation of the Hessian $H(x) = A^T S_x^{-2} A$ of the logarithmic barrier.
To this end, we apply \cref{thm:quantum-approx} to the matrix $S_x^{-1} A$, which gives us an approximation $\widetilde H \approx H(x)$ while making only $\tO(\sqrt{nd}/\eps)$ row queries to $S_x^{-1} A$.
Recalling that $S_x$ is a diagonal matrix with entries $(S_x)_{ii} = (Ax-b)_i$, it is clear that we can simulate a row query to $S_x^{-1} A$ with a single query to $b$ and a single row query to $A$.

For the volumetric barrier and the Lewis weight barrier, we need query access to the rescaled matrices $\Sigma_x^{1/2} S_x^{-1} A$ and $W_x^{1/2} S_x^{-1} A$, respectively.
The rescalings by $\Sigma_x$ and $W_x$ are based on leverage scores and Lewis weights, respectively, and we have to describe additional quantum algorithms for approximating these. Our algorithm to approximate leverage scores is a direct consequence of \cref{thm:quantum-approx}: if we have a spectral approximation $\wt B$ such that $\wt B^T \wt B \approx_\eps B^T B$, we also get that
\[
b_i^T (\wt B^T \wt B)^+ b_i
\approx_\eps b_i^T (B^T B)^+ b_i = \sigma_i(B).
\]
Hence, we can use a spectral approximation of $B$ to estimate its leverage scores. Lewis weights are more complicated to define and approximate, we  describe how to do so in the next section. 
Based on these algorithms we obtain the complexities for Hessian approximation in \cref{tab:hessian-complexities}.

\begin{table}[ht!]
    \centering
    \begin{tabular}{l|l|l}
         Barrier &  Row queries & Time \\ \hline  

         \vspace{-0.75em} &&\\  
         Logarithmic barrier  & $\sqrt{nd}$ & $r \sqrt{nd} + d^\omega$ \\
       Volumetric barrier & $\sqrt{nd}$ & $r \sqrt{nd} + d^{\omega}$ \\
       Lewis weight barrier & $\sqrt{n} d^{7/2}$ & $d^{9/2}(r^2 \sqrt{nd} + d^\omega)$
    \end{tabular}
    \caption{Complexity of quantum algorithms for $\tO(1)$-spectral approximations of Hessians. Row queries are to $A \in \R^{n \times d}$ and $b \in \R^n$, $r$ is the row-sparsity of $A$, $\omega$ is the matrix multiplication coefficient, and we ignore polylog-factors.}
    \label{tab:hessian-complexities}
\end{table}

\subsubsection{Lewis weights} \label{sec:leverage scores and lewis weights}

To approximate Lewis weights, we first note that leverage scores and spectral approximation are implicitly defined with respect to the $\ell_2$-norm: $\wt B$ is a spectral approximation of $B$ if and only if $\| \wt B x \|_2 \approx \| B x \|_2$ for all vectors $x$.
If we extend this notion from the $\ell_2$-norm to the $\ell_p$-norm, we naturally arrive at the notion of $\ell_p$-Lewis weights $\{w^{(p)}_i(B)\}$ of a matrix $B$~\cite{Lewis78,cohen2015lp}.
As is shown in~\cite{cohen2015lp}, these can be used to obtain an $\ell_p$-approximation $\wt B$ of $B$ in the sense that $\| \wt B x \|_p \approx \| B x \|_p$ for all~$x$, and hence they find use in algorithms for $\ell_p$-regression, see also \cite{JLS22}. 
Lewis weights are moreover related to computational geometry via their connection to the John ellipsoid~\cite{YinTatThesis,todd2016minimum}. 
Here we use a connection between Lewis weights and fast linear programming algorithms that was recently established by Lee and Sidford~\cite{LeeSidford19}, see also, e.g.,~\cite{brand2020solving,brand2021minimum,Space2023} for follow-up convex optimization algorithms based on Lewis weights. 
The $\ell_p$-Lewis weights are implicitly defined through a fixed-point equation: they are defined as the leverage scores of a rescaling of $B$ based on the Lewis weights,
\begin{equation} \label{eq:Lewis weights intro}
w^{(p)}_i(B)
= \sigma_i(W^{1/2-1/p} B),
\end{equation}
where $W$ is the diagonal matrix with entries $(W)_{ii} = w^{(p)}_i(B)$.
While more challenging than leverage scores, classical algorithms for approximating the Lewis weights have been described in \cite{cohen2015lp,YinTatThesis,LeeSidford19,Fazel:highprecisionLewisWeights,apers2024lewis}. (For $p \in (0,4)$ the algorithm of~\cite{cohen2015lp} is in fact relatively simple, but here we need $p$ roughly $\log(n)$.) 
These algorithms typically obtain the Lewis weights through an iterative process of computing (approximate) regular leverage scores, and rescaling rows based on these scores.
Specifically, we follow the low-precision algorithm from Apers, Gribling and Sidford \cite{apers2024lewis} (which is a variation on Lee's algorithm \cite{YinTatThesis}).
Effectively, this algorithm returns a multiplicative $\eps$-approximation of the $\ell_p$-Lewis weights by doing $\poly(d)/\eps$ many leverage scores computations, each to precision $\eps/\poly(d)$.
Notably, since we are aiming for a sublinear runtime scaling as $\sqrt{n} \poly(d)$, we have to be more careful than in previous algorithms since we cannot afford to even write down all $n$ Lewis weights.
Instead, we construct implicit data structures that allow us to (relatively) efficiently query the intermediate weights.
We obtain the following theorem, which we believe is of independent interest.\footnote{For the important special case of leverage scores ($p=2$) we obtain significantly better bounds on the complexity, see \cref{lem:direct-LS,lem:approx-LS}.}
\begin{theorem*}[Quantum Lewis weights, informal version of~\cref{thm:quantum-Lewis}]
Consider query access to a matrix $B \in \R^{n \times d}$ with row sparsity~$r$.
For any $0 < \eps \leq 1$ and $p \geq 2$, there is a quantum algorithm that provides query access to $\eps$-multiplicative approximations of the Lewis weights $w^{(p)}_i(B)$.
The algorithm has a preprocessing phase that makes $\tO(\sqrt{n} d^{7/2}/\eps^2)$ row queries to $A$ and takes time $\sqrt{n} \cdot \poly(d,\log(n),1/\eps)$.
After this, each query requires 1 row query to $B$ and time $\tO(r^2 d^{3/2}/\eps)$.
\end{theorem*}

\subsubsection{Quantum gradient and matrix-vector approximation}

We now turn to the approximation of the gradient, which we need to approximate in the ``inverse-local norm'', i.e., find $\tilde g$ such that $\| \tilde g - g \|_{H^{-1}} \leq \delta$.
We first note that, if $B^T B \preceq H$ then $H^{-1} \preceq (B^T B)^{-1}$, and therefore it suffices to find approximations of the gradient in the $(B^T B)^{-1}$-norm where $B$ is as in the previous section. An important observation is that, for these~$B$'s, each of the gradients in \cref{tab:self-concordant barriers} take the form $g = B^T v$ where $v$ is either $\mathbf 1$, $\Sigma_x^{1/2} \mathbf 1$, or $W_x^{1/2} \mathbf 1$.  Our goal is thus, given query access to $B$ and $v$, to find a vector $\tilde y$ such that 
\[
\|\tilde y - B^T v\|_{(B^TB)^{-1}} \leq \delta.
\]
To do so, we use a recent quantum algorithm for multivariate mean estimation by Cornelissen, Hamoudi and Jerbi~\cite{CHJ:multivariate}.
For a random variable $Y$ with mean $\mu$ and covariance matrix $\Sigma$, this returns an estimate $\tilde\mu$ satisfying $\| \tilde\mu - \mu \|_2 \leq \delta$ while using only $\tO(\sqrt{d \Tr(\Sigma)}/\delta)$ samples of $Y$ (as compared to $O(\Tr(\Sigma)/\delta^2)$ samples that are required classically).\footnote{The quantum bound stated here assumes that $\Tr(\Sigma)/\delta^2 \geq d$, which will be satisfied in our applications.}

It is natural to rewrite $g = B^T v$ as the mean $\mu$ of a random variable $Y = n v_i B^T e_i$, with $i \in [n]$ uniformly at random, and then apply quantum mean estimation to $Y$.
However, it is not clear how to control the trace of the corresponding covariance matrix $\Sigma$, one can only show $\Sigma \preceq n B^T \diag(v)^2 B$, and naively turning the guarantee on $\| \tilde\mu - \mu \|_2$ into the required guarantee on $\| \tilde \mu - \mu \|_{H^{-1}}$ would introduce an unwanted dependence on the condition number of $H$.

We avoid this problem by first constructing a spectral approximation $\widetilde H \approx_\delta H$ using \cref{thm:quantum-approx}.
This allows us to sample from the ``pre-conditioned'' random variable $Y = n v_i \widetilde H^{-1/2} B^T e_i$ with $i \in [n]$ uniformly at random.
Now note that $\Sigma = \E[Y Y^T] = n \widetilde H^{-1/2}B^T  B \widetilde H^{-1/2}\approx n I_d$, and moreover
\[
\delta
\geq  \| \tilde \mu - \mu \|_2
 = \| \widetilde H^{1/2} \tilde\mu - B^T v \|_{\widetilde H^{-1}}
\approx \| \widetilde H^{1/2} \tilde\mu - B^T v \|_{H^{-1}},
\]
so that returning the estimate $\widetilde H^{1/2} \tilde\mu$ yields a valid estimator of the gradient.
This trick avoids any condition number dependence, and it leads to the following theorem.

\begin{restatable*}[Approximate matrix-vector product]{theorem}{approxmv} \label{thm:approx-mv}
Assume query access to a vector $v \in \R^{n}$ and a matrix $B \in \R^{n \times d}$ with row sparsity $r$.
There is a quantum algorithm that returns a vector $\tilde y$ satisfying
\[
\| \tilde y - B^T v \|_{(B^T B)^{-1}}
\leq \delta,
\]
while making {$\tO(\sqrt{n}d\|v\|_\infty/\delta)$} row queries to $B$ and $v$, and taking time {$\tO(r \sqrt{n}d^{2}\|v\|_\infty/\delta + d^\omega)$}.
\end{restatable*}

Applying this theorem to the various barrier functions, we obtain the complexities in \cref{tab:gradient-complexities} below.

\begin{table}[ht!]
    \centering
    \begin{tabular}{l|l|l}
         Barrier &  Row queries & Time \\ \hline 
       Logarithmic barrier  & $\sqrt{n}d$ & $r\sqrt{n}d^2 + d^\omega$ \\
       Volumetric barrier & $\sqrt{n}d$ & $r\sqrt{n}d^2 + d^\omega + d \cdot \min\{d^\omega,d r^2\}$ \\
       Lewis weight barrier & $\sqrt{n}d^{9/2}$ & $r^2\sqrt{n}d^{13/2} + d^{\omega+13/2}$
    \end{tabular}
    \caption{Complexity of quantum algorithms for estimating gradients up to constant error. Row queries are to $A \in \R^{n \times d}$ and $b \in \R^n$, $r$ is the row-sparsity of $A$, and we ignore polylog-factors.}
    \label{tab:gradient-complexities}
\end{table}

\subsection{Main result and discussion} \label{sec:main-discussion}

We summarize our main result, which is based on the Lewis weight IPM.
Combining our quantum algorithms for spectral approximation, Lewis weight computation, and approximate matrix-vector multiplication, we obtain the following theorem.

\begin{theorem}[Quantum IPM for LP] \label{thm:quantum-IPM}
Consider an LP\footnote{We assume that the LP is full-dimensional, and that we are given an initial point $x_0$ in the interior of the feasible region, not too close to the boundary. See \cref{sec:robustness} for a more precise statement.} $\min\, c^T x$ s.t.~$Ax \geq b$.
Assume $A \in \R^{n \times d}$ has row sparsity $r$, and let $\mathrm{val}$ denote its optimal value.
Given query access to $A,b,c$, we give a quantum interior point method that explicitly returns a vector $\tilde x$ satisfying
\[
A \tilde x \geq b
\quad \text{ and } \quad
c^T \tilde x \leq \mathrm{val} + \eps.
\]
Depending on the barrier function, the algorithm makes $\tO(n d)$ (logarithmic), $\tO(n^{3/4} d^{5/4})$ (volumetric) or $\tO(\sqrt{n} d^5)$ (Lewis weight) row queries to $A$, $b$ and $c$.
For the Lewis weight barrier, the time complexity is $\tO(\sqrt{n} \poly(d))$. 
\end{theorem}

The quantum IPMs based on the volumetric and Lewis weight barriers use a sublinear number of row queries to $A$ when $n$ is much larger then $d$ (specifically, $n \in \Omega(d^5)$ for volumetric and $n \in \Omega(d^{10})$ for Lewis weights).
Such sublinear quantum query and time complexity is in line with many earlier works on quantum speedups such as \cite{durr2006quantum,apers2022quantum} but also works on classical sublinear algorithms \cite{rubinfeld2011sublinear}.
In the following we put our results in perspective by (i) discussing such tall LPs, (ii) describing a lower bound, and (iii) describing a benchmark algorithm based on cutting plane methods.

\paragraph{Tall LPs.}
Our quantum algorithm has sublinear query and time complexity for LPs with many more constraints than variables.
Such ``tall'' LPs have a long history, see for instance the early works~\cite{megiddo1984fixed,clarkson1995vegas} as well as the recent~\cite{brand2021minimum}. 
While there are many examples of such LPs e.g.~in combinatorial optimization \cite{schrijver2003combinatorial}, we mention one concrete application mentioned in the early work~\cite{megiddo1984fixed}.
``Linear separability'' asks whether two sets of $O(n)$ vectors in $\R^d$ can be separated by a hyperplane. A special case of this problem is called the linear separability problem of Boolean functions: here the set of points consists of the $n=2^d$ vertices of the $d$-dimensional Boolean hypercube, which are divided into two sets based on the value of a Boolean function.
Boolean functions that admit a linear separator are also known as perceptrons.
The resulting LP has exponentially many more constraints than variables.
Moreover, if the Boolean function can be evaluated efficiently then the constraints have a \textit{concise} description, and so this is an example where the queries in \cref{thm:quantum-IPM} can be efficiently instantiated.

\paragraph{Lower bound.}
The lower bound essentially follows from~\cite{vAGGdW:quantumSDP}, and we present it in \cref{sec: LB}.
The bound states that any quantum algorithm for solving an LP to constant precision must make $\Omega(\sqrt{nd})$ row queries.
This shows that the $\sqrt{n}$-dependency in our algorithm is optimal, while it leaves much room for improvement regarding the $d$-dependency.
Indeed, \emph{we conjecture that this lower bound is tight,} and we later give some hopeful directions to improve our results towards such a bound. For comparison, in \cref{sec: LB} we show that in the classical setting $\Omega(n)$ queries are needed, either to the same row query model, or to a quantum-inspired access model.

\paragraph{Benchmark.}
We now compare our results to a quantum speedup over \emph{cutting plane methods}.
In \cref{app: cutting plane} we show that such methods are particularly amenable to a Grover-type quantum speedup.
Indeed, they are usually phrased in terms of a ``separation oracle'' that returns an arbitrary violated constraint, and this can be implemented quadratically faster using Grover search.
The state-of-the-art (and highly technical) cutting plane method from Lee, Sidford and Wong~\cite{LeeSidfordWong15} makes $\tO(d)$ queries to a separation oracle, and this directly yields a quantum algorithm that makes $\tO(\sqrt{n} d)$ row queries to $A$ (and requires additional time $\tO(r \sqrt{n} d + d^3)$).
The query complexity is a factor $\sqrt{d}$ above the lower bound, and such a black-box approach building on a classical algorithm with separation queries cannot further improve over this (indeed, it is not hard to see that any classical algorithm for LP solving must make $\Omega(d)$ separation queries).
We see this as a clear motivation for looking at ``white-box'' quantum algorithms, such as ours based on IPMs. Indeed, also classically the state-of-the-art for solving tall, dense, LPs is achieved by an IPM~\cite{brand2021minimum}; it runs in time~$\tO(nd+d^{2.5})$.

\paragraph{Open questions.}
Finally, we present some of the open directions to improve over our work. The first two relate to the complexity of approximating a single Newton step.

\begin{enumerate}
\item[1.]
\textbf{Cheaper gradient approximation:}
We can spectrally approximate the Hessian with $\tO(\sqrt{nd})$ row queries for the logarithmic and the volumetric barrier.
This is tight, and it matches the lower bound for LP solving.
However, our algorithm for approximating the gradient of the simplest, logarithmic barrier already has complexity $\tO(\sqrt{n} d)$. 
This is a clear direction to improve the complexity of our quantum IPM.
To surpass our bound, one could try to better exploit the structure of the gradients that we need to estimate, or change the path-following method so that a weaker approximation of the gradient would suffice.

\item[2.] \textbf{Low-precision leverage scores / Lewis weights:} Currently, the (time) complexity of approximate Newton steps for the volumetric and Lewis weight barriers is higher than that of the logarithmic barrier. The main reason for this is that, in order to obtain a good approximation of the gradient, we currently require high-precision $1/\poly(d)$-approximations of leverage scores. For Lewis weights in particular, the cost of converting between different notions of approximations is currently high and thus a clear candidate for improvement, see \cite{apers2024lewis,Fazel:highprecisionLewisWeights}. Can one base an IPM with these barriers solely on $O(1)$-approximate leverage scores? 
\end{enumerate}

While the cost of a single Newton step could conceivably be improved to $\tO(\sqrt{nd})$ row queries, an IPM still requires a large number of such steps (e.g., $\tO(\sqrt{d})$ Newton steps for the Lewis weight barrier). 
Dealing with this overhead is another natural direction.
\begin{enumerate}
\item[3.]
\textbf{Dynamic data structures:} Classical state-of-the-art algorithms often use dynamic data structures to \emph{amortize} the cost of an expensive routine (e.g., computing a Newton step) over the number of iterations. Maintaining such a data structure in a sublinear time quantum algorithm is challenging and has not been explored much. One nice example is \cite{bouland2023quantum} where the authors introduce dynamic Gibbs sampling to improve the state-of-the-art first-order method for zero-sum games (and thereby also for LPs).
IPMs also seem a natural setting in which to explore this question. For example, the Hessian changes slowly between iterations of an IPM by design. Whereas we show that the complexity of our spectral approximation algorithm is tight when applied to a single $B^T B$, it seems likely that this can be improved when considering a sequence of slowly changing matrices.
\end{enumerate}
There are some unexplored applications of our quantum algorithms.
\begin{enumerate}
\item[4.]
\textbf{Further applications:}
As a subroutine, we derived quantum algorithms for estimating leverage scores and Lewis weights.
These quantities have their independent merit, e.g.~for doing dimensionality reduction in data science.
Indeed, following up on a first version of our paper, the works \cite{song2023revisiting} and \cite{gao2024quantum} already built on our quantum leverage score algorithm to derive new quantum algorithms for linear regression and spectral approximation of Kronecker products.

More in the direction of randomized linear algebra, the concept of \emph{determinantal point processes} measure the importance of larger subsets of rows.
Recent works studied quantum algorithms for these based on quantum linear algebra \cite{kerenidis2022quantummachinelearningsubspace,Bardenet2024determinantal}, and it would be interesting to see if our new techniques can be applied to this setting.
\end{enumerate}

Finally, in our work we established a quantum speedup for tall LPs, i.e., when $n \gg d$. This raises two natural questions. First, we achieved a sublinear number of row-queries when $n \gg d^{10}$ using the Lewis weight barrier, or $n \gg d^{5}$ using the volumetric barrier. What is the smallest exponent $c$ for which we can achieve a sublinear number of row queries when $n \gg d^c$? Second, more broadly, can one achieve a speedup in the regime where $n \approx d$?

\subsection{Related work}

\paragraph{Quantum interior point methods.}
As one of the motivations for this work, we mention that a large number of works \cite{kerenidis2020quantum,kerenidis2021quantum,augustino2021quantum,mohammadisiahroudi2022efficient,dalzell2022end,huang2022faster} have recently studied the use of quantum linear system solvers  combined with tomography methods to speed up the costly Newton step at the core of an IPM. We refer to the recent survey~\cite{Abbas2024Challenges} for a more detailed discussion of the challenges quantum linear system solvers introduce in the context of IPMs. In particular, quantum linear system solvers inherently introduce a dependence on the condition number of the linear system that they solve. A main feature of our work is that we do not introduce such a dependence. A second main feature of our work is that we achieve a $\log(1/\eps)$-dependence, whereas previous quantum algorithms for LPs usually achieved a $\poly(1/\eps)$-dependence~\cite{brandaoSDP17,vAGGdW:quantumSDP,brandão2019quantum,vanapeldoornImprovedSDP,vanapeldoorn2019quantum,bouland2023quantum}.
Recently, iterative refinement techniques have been used in the context of quantum IPMs to improve the error dependence to $\polylog(1/\eps)$, but the dependence on condition numbers remains~\cite{mohammadisiahroudi2022efficient,augustinoThesis}.

We finally mention a recent proposal to solve LPs by following the central path through quantum simulation of suitable dynamics~\cite{augustino2024quantumcentralpathalgorithm}. The claimed query complexity is $\widetilde O(\sqrt{n+d} \cdot \nnz(A) \cdot R_1/\eps)$ where $\nnz(A)$ is the number of non-zero entries in $A$ and $R_1$ is an upper bound on the $\ell_1$-norm of the solutions to both the primal and dual problem. We note that this upper bound scales at least with $n^{1.5}$ (assuming every constraint is non-zero) and also exhibits a $1/\eps$-scaling.

\paragraph{Quantum gradient and Hessian approximation.} In a different setting, where one can make function-evaluation queries, Jordan introduced an algorithm to estimate gradients~\cite{jordan2005fast}. This algorithm was later rigorously analyzed for smooth functions~\cite{gilyen2019optimizing}, convex functions~\cite{vAGGdW:quantumconvex,Chakrabarti2020quantumalgorithms}; and extended to Hessian estimation~\cite{zhang2024quantumspectralmethodgradient}. These algorithms differ from our setting in both their input and output models. As for the input model, they require quantum query access to an oracle that evaluates the function (up to high precision); implementing such a query for, e.g., the logarithmic barrier naively requires querying the entire input. As for the output model, both for the gradient and Hessian approximation one obtains an entrywise approximation. In contrast, our algorithm requires a spectral approximation of the Hessian, for which such an additive approximation is insufficient.

\paragraph{Quantum linear algebra.} 
A recent line of work has unified many quantum algorithms for linear algebraic problems, using a framework called quantum singular value transformation (QSVT)~\cite{gilyen2019qsvt}. This includes for instance the well-known quantum algorithm for solving linear systems~\cite{chakraborty2018powerc}.
We mention two recent works that are most related to ours. First, the QSVT-framework has been used to construct a leverage score sampling algorithm in~\cite{shao2023quantumspeedupleveragescore}. The runtime of their algorithm depends on the conditioning of the matrix, in contrast to the algorithm that we develop here. Second, the QSVT-framework can be used to approximate the top eigenvector of a $d$-by-$d$ matrix $A$ in sublinear time given query access to the entries of $A$~\cite{chen2025topeigenvector}. Their algorithm implements a robust version of the power method, and thus establishes an efficient algorithm for matrix-vector multiplication for square matrices. We develop an efficient algorithm for matrix-vector multiplication in the regime where the matrix is tall. The two methods moreover measure the quality of the approximation in a different norm; they use the standard $\ell_2$-norm, whereas we use a norm motivated by our application in IPMs (the dual of the local norm).

\paragraph{Quantum-inspired algorithms.}
Finally, our spectral approximation for tall matrices $B \in \R^{n \times d}$ with $n \gg d$ is based on subsampling the rows of $B$.
This setting is much in line with recent works on quantum-inspired classical algorithm (see e.g.~Tang~\cite{tang19recommendation} and Chia et al.~\cite{chia22sampling}).
However, a critical difference is that these works sample rows of $B$ based on the $\ell_2$-norms of the rows, and this typically yields additive approximations.
In our work we sample rows of $B$ based on the $\ell_p$-leverage scores.
This yields multiplicative approximations, and this is crucial for our IPM setting.
We formalize this point by showing in \cref{sec: LB} an $\Omega(n)$ lower bound for spectral approximation and LP solving in the quantum-inspired access model.

\section{Preliminaries}

\paragraph{Notation.}
We let $\mathcal S^d \subseteq \R^{d \times d}$ be the set of $d$-by-$d$ symmetric matrices. For $H \in \mathcal S^d$ we write $H \succ 0$ (resp.~$H\succeq 0$) if $H$ is positive definite (resp.~positive semidefinite). For vectors $u, v \in \R^d$ we let $\langle u, v \rangle = \sum_{i \in [d]} u_i v_i$ denote the standard inner product. A positive definite matrix $H$ defines an inner product $\langle u,v \rangle_H \coloneqq \langle u,Hv\rangle$. We write $\|u\|_H = \sqrt{\langle u,u\rangle_H}$. We often use the \emph{local norm} associated to a twice continuously differentiable function $f$ whose Hessian $H(x)$ is positive definite for all $x$ in its domain $D$. Given a fixed reference inner product $\langle \cdot, \cdot \rangle$ (e.g.~the standard inner product on $\R^d$), $f$ defines a family of \emph{local} inner products, one for each $x$ in its domain via its Hessian:
\[
\langle u,v\rangle_{H(x)}
\coloneqq \langle u,H(x)v\rangle. 
\]
We will write $\langle \cdot, \cdot \rangle_x$ instead of $\langle \cdot, \cdot \rangle_{H(x)}$ when the function $f$ is clear from context.  

For $\eps>0$, the matrix $Q \in \S^d$ is an $\eps$-spectral approximation of a $d$-by-$d$ PSD matrix~$H \succeq 0$ if 
\[
(1-\eps) Q \preceq H \preceq (1+\eps) Q, 
\]
where $A \preceq B$ is equivalent to $B-A \succeq 0$.
We denote such a spectral approximation by $Q \approx_{\eps} H$.
We will also use this notation for scalars, where $a \approx_\eps b$ denotes $(1-\eps) b \leq a \leq (1+\eps) b$, and for vectors (where the inequalities hold entrywise).

Finally, for a vector $v \in \R^n$, we use the notation $V = \Diag(v)$ for the diagonal $n$-by-$n$ matrix who's $i$-th diagonal element corresponds to the $i$-th element of $v$.

\paragraph{With high probability.}
Throughout this work, ``with high probability'' means with probability at least $1-1/n^c$ for an arbitrarily high but fixed constant $c>0$, where $n$ typically denotes the size of the problem instance.

\paragraph{Quantum computational model and QRAM assumptions.}
We assume the usual quantum computational model (see e.g.~\cite{apers2022quantum}), which is a classical system that can (i) run quantum subroutines on $O(\log n)$ qubits, (ii) can make quantum queries to the input, and (iii) has access to a quantum-read/classical-write RAM (QRAM) of $\poly(n)$ bits.\footnote{We charge unit cost for classically writing a bit, or quantumly reading a bit (potentially in superposition).}
The \emph{time complexity} of an algorithm in this model measures the number of elementary classical and quantum gates, queries and QRAM operations that the algorithm uses.
The \emph{query complexity} of an algorithm measures the number of queries to the input.

The feasibility of a QRAM is often debated.
However, we note that by a standard argument (see, e.g., \cite[Theorem 1]{yuan2023optimal}) one can simulate access to a QRAM of $k$ bits by increasing the number of gates and qubits in the regular circuit model, and hence the algorithm's time complexity, by a multiplicative factor $O(k)$ (and only a $O(\log(k))$ increase in circuit depth).
As a consequence, all statements about quantum \emph{query} complexity are \emph{independent of the QRAM assumption}.
This is not the case for the time complexity, however, since the algorithm from \cref{lem:q-sampling} uses a QRAM of size $\sqrt{n}\, \poly(d)$. 

Regarding queries to the input, we will typically consider inputs consisting of tall matrices $Z \in \R^{n \times d}$ (where $d$ can be 1).
We count the query complexity in terms of the number of \emph{row queries}, where a row query to the $i$-th row of $Z$ returns the entire row.
We note that a query can be made in superposition.
The row query complexity can alternatively be expressed in the \emph{sparse access} model (see e.g.~\cite[Section 2.4]{chakraborty2018powerc}) by noting that a single row query corresponds to $r$ sparse queries, with~$r$ the row-sparsity of $Z$.
We will frequently work with matrices of the form $D Z$ for a diagonal matrix $D$, and we note that one row query to $D Z$ can be reduced to 1 query to $D$ and 1 row query to $Z$.

\section{Quantum algorithm for spectral approximation}

In this section we describe our main quantum algorithm for constructing a spectral approximation of a matrix $A \in \R^{n \times d}$ (where now $A$ denotes a general matrix, not necessarily the LP constraint matrix).
The algorithm has sublinear time and query complexity for sufficiently tall matrices.
While motivated by interior point methods, it is a self-contained section, and we envision applications of these quantum algorithms in other areas such as statistical regression.
The algorithm's complexity is described in the following theorem.

\begin{theorem}[Quantum spectral approximation] \label{thm:quantum-approx}
Consider query access to a matrix $A \in \R^{n \times d}$ with row sparsity $r$.
For any $0 < \eps \leq 1$, there is a quantum algorithm that, with high probability, returns a matrix $B \in \R^{\tO(d/\eps^2) \times d}$ satisfying
\[
(1-\eps) B^T  B
\preceq A^T A
\preceq (1+\eps) B^T B,
\]
while making $\tO(\sqrt{nd}/\eps)$ row queries to $A$, and taking time $\tO(r \sqrt{nd}/\eps + d^\omega)$.
\end{theorem}

Based on this theorem, and with some additional work, we can derive the following theorem on quantum algorithms for approximating leverage scores.
We will use this later on.

\begin{restatable}[Quantum leverage score approximation]{theorem}{quantumls}
\label{thm:quantum-ls}
Assume query access to a matrix $A \in \R^{n \times d}$ with row sparsity $r$.
For any $0 < \eps \leq 1$, there is a quantum algorithm that, with high probability, provides query access to estimates $\tilde\sigma_i$ for any $i \in [n]$ satisfying
\[
\tilde\sigma_i
= (1\pm \eps) \sigma(A).
\]
The cost of the algorithm is either of the following:
\begin{enumerate}
\item
Preprocessing: $\tO(\sqrt{nd}/\eps)$ row queries to $A$ and $\tO(r \sqrt{nd}/\eps + d^\omega + \min\{d^\omega,d r^2\}/\eps^2)$ time.\\
Cost per estimate $\tilde\sigma_i$: 1 row query to $A$ and $O(r^2)$ time.
\item
Preprocessing: $\tO(\sqrt{nd}/\eps)$ row queries to $A$ and $\tO(r \sqrt{nd}/\eps + d^\omega + \min\{d^\omega,d r^2\}/\eps^2 + d^2/\eps^4)$~time.\\
Cost per estimate $\tilde\sigma_i$: 1 row query to $A$ and $\tO(r/\eps^2)$ time.
\end{enumerate}
\end{restatable}

Proofs of both theorems are given in \cref{sec:q-rep-halving}.

\begin{remark}[Graph sparsification] \label{remark:graph-sparsification}
We note that \cref{thm:quantum-approx} generalizes an earlier quantum algorithm for graph sparsification \cite{apers2022quantum}, although it uses very different techniques.
Indeed, we can recover the complexity in \cite[Theorem 1]{apers2022quantum} by letting $A \in \R^{n \times d}$ denote the edge-vertex incidence matrix of a graph with $n$ edges and $d$ vertices, which has sparsity $r = 2$.
This directly recovers the $\tO(\sqrt{nd}/\eps)$ query complexity of \cite{apers2022quantum}.
To also recover the $\tO(\sqrt{nd}/\eps)$ time complexity of \cite{apers2022quantum}, note that the $d^\omega$ term in \cref{thm:quantum-approx} is due to solving linear systems for matrices of the form $B^T B \in \R^{d \times d}$, where $B$ contains a (rescaled) subset of $\tO(d)$ rows of~$A$.
When $A$ is an edge-vertex incidence matrix, then $B^T B$ will be a Laplacian matrix with $\tO(d)$ nonzero entries, and these can be solved in near-linear time $\tO(d)$ using fast Laplacian solvers \cite{spielman2004nearly}. 

From the graph setting, we get a matching lower bound that also applies to our setting: in \cite[Theorem~2]{apers2022quantum} it was shown that $\wt\Omega(\sqrt{nd}/\eps)$ quantum queries are needed to construct $\eps$-spectral sparsifiers.
See \cref{lemma:spectral-LB} for a more detailed statement.
\end{remark}

\subsection{Repeated halving algorithm} \label{sec:rep-halving}

We base our quantum algorithm for spectral approximation on a simple, elegant, recursive sparsification algorithm described by Cohen, Lee, Musco, Musco, Peng and Sidford in~\cite{cohen2015uniform}.
As most sparsification algorithms, it requires subsampling rows.
The simplest way of subsampling is just keeping any individual row with a probability $q$.
We denote by
\[
B \subseteq_q A
\]
the matrix obtained by keeping every individual row of $A$ independently with probability $q$.
We will also use a more refined notion of subsampling.
For a matrix $A \in \R^{n \times d}$ and an entrywise positive weight vector $w = (w_i)_{i \in [n]}$, we denote by
\[
B \overset{w,\epsilon}{\leftarrow} A
\]
the matrix obtained by the following weighted subsampling process (where $c$ is a large enough universal constant):

\begin{algorithm}[H]
\caption{Weighted subsampling $B \overset{w,\epsilon}{\leftarrow} A$} \label{alg:subsampling}
\Input{$A \in \R^{n\times d}$, positive $w \in (\R_{>0} \cup \infty)^n$ and $\eps>0$}

\BlankLine

\For{$i \in [n]$}{ 
    with probability $q_i \coloneqq \min\{1,c w_i \log(d)/\epsilon^2\}$, add row $\frac{1}{\sqrt{q_i}} a_i^T$ to $B$\; 
}
\Return{$B$}
\end{algorithm}

The rescaling ensures that, irrespective of $w$, we have that if $B \overset{w,\epsilon}{\leftarrow} A$ then
\[
\E[B^T B]
= \sum_i q_i \frac{1}{q_i} a_i a_i^T
= \sum_i a_i a_i^T
= A^T A.
\]
We can ensure concentration of $B^T B$ around its expectation by picking the weights appropriately.
A typical choice is based on \emph{leverage scores}.
The leverage score of the $i$-th row $a_i^T$ of a matrix $A \in \R^{n \times d}$ is defined as
\[
\sigma_i(A)
\coloneqq 
a_i^T (A^T A)^+ a_i.
\]
The following lemma demonstrates that sampling probabilities based on leverage scores ensure concentration. \begin{lemma}[{Consequence of the matrix Chernoff bound \cite{Tropp11:matrixconcentration}, \cite[Lemma 4]{cohen2015uniform}}] \label{lem:matrix-chernoff}
Consider matrix $A \in \R^{n \times d}$ and weight vector $w = (w_i)_{i \in [n]}$ satisfying $w_i \geq \sigma_i(A)$.
Then, with high probability,
$B \overset{w,\epsilon}{\leftarrow} A$ satisfies $B^T B \approx_\epsilon A^T A$ and has at most $O(\sum_{i \in [n]} q_i)$ rows, where we recall that $q_i =  \min\{1,c w_i \log(d)/\epsilon^2\}$.
\end{lemma}

Now note that we always have $\sum_{i \in [n]} q_i \leq c\log(d) \|w\|_1/\epsilon^2$ and
\[
\sum_i \sigma_i(A)
= \sum_i a_i^T (A^T A)^+ a_i
= \Tr(A^T A (A^T A)^+)
\leq d.
\]
Hence, if we have good estimates $w_i \in \Theta(\sigma_i(A))$ of the leverage scores of $A$, then $\|w\|_1 \in O(d)$ and the resulting matrix $B$ has only $O(d \log(d)/\eps^2)$ rows.
Of course, in general it is not clear how to efficiently obtain such good estimates of the leverage scores. We will follow the approach of~\cite{cohen2015uniform}, which uses subsampling to cheaply compute overestimates of leverage scores. To formalize this, we need a slight generalization of leverage scores, also used in~\cite{cohen2015uniform}: the \emph{generalized} leverage score of the $i$-th row of a matrix $A \in \R^{n \times d}$ with respect to a matrix $B \in \R^{n_B \times d}$ is defined by
\[
\sigma^B_i(A)
\coloneqq \begin{cases}
a_i^T (B^T B)^+ a_i &\text{if } a_i \perp \ker(B) \\
\infty &\text{otherwise}.
\end{cases}
\]
Note that (i) $\sigma_i(A) = \sigma^A_i(A)$ and (ii) if $B \subseteq_q A$ for $q \in [0,1]$ then $\sigma^B_i(A) \geq \sigma_i(A)$.
The second statement shows that we can derive overestimates of the leverage scores from a uniformly subsampled matrix $B$ of $A$.
The following theorem from~\cite{cohen2015uniform} shows that if $q = 1/2$ then these overestimates do not increase the resulting number of sampled rows significantly.

\begin{theorem}[{\cite[Theorem 4]{cohen2015uniform}}] \label{thm:uniform}
Consider $A \in \R^{n \times d}$ and $B \subseteq_{1/2} A$.
Define $w \in \R^n$ by $w_i = \sigma^B_i(A)$.
Then, with high probability, $\sum_{i \in [n]} q_i \in O(d \log(d)/\eps^2)$, where we recall that $q_i =  \min\{1,c w_i \log(d)/\epsilon^2\}$.
Consequently, if $\widetilde B \overset{w,\eps}{\leftarrow} A$, then with high probability $\widetilde B^T \widetilde B \approx_\eps A^T A$ and $\widetilde B$ has $O(d \log(d)/\eps^2)$ rows.
\end{theorem}

Of course, the uniform subsampling $A' \subseteq_{1/2} A$ only reduces the matrix size by a factor roughly two.
In the following algorithm (based on~\cite[Algorithm 2]{cohen2015uniform}), a chain of $L \in O(\log(n/d))$ uniformly downsampled matrices $A_1,\dots,A_L$ is considered, which does result in a matrix $A_L$ with only $\tO(d)$ rows.
By \cref{thm:uniform} we can then derive leverage scores from (an approximation $B_\ell$ of) $A_\ell$ to construct an approximation $B_{\ell-1}$ of $A_{\ell-1}$.
Repeating this process $L$ times finally yields an approximation $B_0$ of $A_0 = A$.

\begin{algorithm}[H]
\caption{Repeated halving algorithm~\cite{cohen2015uniform}} \label{alg:repeated-halving}
\Input{matrix $A \in \R^{n \times d}$, approximation factor $\eps>0$}

\BlankLine

let $A_L \subseteq_{1/2} \dots \subseteq_{1/2} A_1 \subseteq_{1/2} A$ for $L = \lceil \log_2(n/d) \rceil$; let $B_L = A_L$\;

\For{$\ell = L-1,L-2,\dots,1$}{ 
    let $B_\ell \overset{w,1/2}{\leftarrow} A_\ell$ with $2\sigma^{B_{\ell+1}}_i(A_\ell) \leq w_i \leq 4\sigma^{B_{\ell+1}}_i(A_\ell)$\; 
}

let $B \overset{w,\epsilon}{\leftarrow} A$ with  $2\sigma^{B_1}_i(A) \leq  w_i \leq 4\sigma^{B_1}_i(A)$\;

\Return{$B$}
\end{algorithm}

\begin{lemma}[Repeated halving] \label{lem:repeated-halving}
Consider \cref{alg:repeated-halving}.
With high probability, every $B_\ell$ has $O(d\log(d))$ rows, and the output $B$ has $O(d \log(d)/\epsilon^2)$ rows and satisfies~$B^T B \approx_\epsilon A^T A$.
\end{lemma}
\begin{proof}
First we prove that for every $1 \leq \ell < L$, with high probability, $B_\ell$ has $O(d \log(d))$ rows and $B_\ell^T B_\ell \approx_{1/2} A_\ell^T A_\ell$.
We prove this by induction.
The base case $\ell = L$ is clear: $B_L = A_L$ and, by a multiplicative Chernoff bound\footnote{Let $Y = \sum_{i=1}^n X_i$, with $X_i$ the random variable indicating whether the $i$-th row of $A$ is in $A_L$. Then $\mu \coloneqq \E[Y] \in \Theta(d)$. The multiplicative Chernoff bound for random variables in $\{0,1\}$ states that $\Pr[|Y-\mu| \geq \delta \mu] \leq 2 e^{-\delta^2 \mu/3}$ for $0 \leq \delta \leq 1$, and so $Y \in \Theta(d)$ except with probability $e^{-\Omega(d)}$.}, $A_L$ has $\Theta(d)$ rows with high probability.
For the inductive step, assume $B_{\ell+1}^T B_{\ell+1} \approx_{1/2} A_{\ell+1}^T A_{\ell+1}$.
Then $\sigma^{B_{\ell+1}}_i(A_\ell) = (1 \pm 1/2) \sigma^{A_{\ell+1}}_i(A_\ell)$, so that the assumption on $w_i$ implies that $\sigma_i^{A_{\ell+1}}(A_\ell) \leq  w_i \leq 6 \sigma_i^{A_{\ell+1}}(A_\ell)$.
Hence we can apply \cref{thm:uniform} with $\epsilon = 1/2$, which implies that with high probability $B_\ell^T B_\ell \approx_{1/2} A_\ell^T A_\ell$ and $B_\ell$ has $O(d \log d)$ rows.

It remains to prove that step 3.~works correctly, assuming that $B_1 \approx_{1/2} A_1$, but this follows by the same argument.
\end{proof}

\subsection{Quantum repeated halving algorithm}

Here we discuss an efficient quantum algorithm for spectral approximation based on the repeated halving algorithm.
This will prove \cref{thm:quantum-approx}. As we have seen, this requires efficient (quantum) algorithms for three steps: (i) a data structure that allows us to efficiently query approximate generalized leverage scores, (ii) a quantum implementation of the procedure $B \overset{w,\eps}{\leftarrow} A$, and (iii) a way to query the downsampled matrices $A_L \subseteq_{1/2} \cdots \subseteq_{1/2} A_1 \subseteq_{1/2} A$. We discuss these three steps in the next three subsections and then combine them to obtain \cref{thm:quantum-approx}.

\subsubsection{Leverage scores via Johnson-Lindenstrauss} \label{sec:JL}

Consider an approximation factor $\eps>0$. 
A bottleneck in the repeated halving algorithm is computing a multiplicative $(1\pm\eps)$-approximation of the generalized leverage scores $\sigma_i^B(A)$, where $A \in \R^{n \times d}$ has $r_A$-sparse rows $a^T_i$, and $B \in \R^{D \times d}$ has $r_B$-sparse rows with $D \geq d$. 
Naively (but exactly) computing the generalized leverage scores requires time $O(r_A^2)$ per leverage score, after a preprocessing cost of $O(d^\omega + \min\{D d^{\omega-1},D r_B^2\})$. 
Here the preprocessing consists of (i) computing the $d \times d$ matrix $B^T B$, which can be done in time $O(D/d \cdot d^\omega)$ by computing $D/d$ matrix products of $d \times d$ matrices, or in time $O(D r_B^2)$ by computing $D$ outer products of $r_B$-sparse vectors, and (ii) computing the pseudo-inverse $(B^TB)^+$, which takes an additional time  $O(d^{\omega})$.
From this we get the following lemma.

\begin{lemma}[Direct leverage scores] \label{lem:direct-LS}
Assume query access to $A \in \R^{n \times d}$, and assume we are explicitly given $B \in \R^{D \times d}$ with $n,D \gg d$.
There is a classical algorithm that, after a precomputation of time $O(d^\omega + \min\{D d^{\omega-1},D r_B^2\})$, with high probability returns estimates $\tilde \sigma_i$ for any $i \in [n]$ that satisfy $\tilde \sigma_i = (1\pm\eps) \sigma_i^B(A)$ at a cost per estimate of one row query to $A$ and time $O(r_A^2)$.
\end{lemma}

Using standard sketching techniques based on the Johnson-Lindenstrauss lemma we can reduce the cost per query. This was used for instance by Spielman and Srivastava~\cite{spielman2011graph} for graph sparsification to query effective resistances. For leverage scores, the construction we use here was also applied e.g.~in \cite[Lemma~7]{cohen2015uniform}. Concretely, we can use the Johnson-Lindenstrauss lemma to get $\eps$-approximate leverage scores at cost $O(r_A \log(n)/\eps^2)$ per leverage score, after a similar precomputation cost. 
We will use the Johnson-Lindenstrauss lemma in the following form.
\begin{lemma}[{Johnson-Lindenstrauss \cite{johnson1984extensions}, \cite[Lemma 1]{larsen2017optimality}}] \label{lem:JL}
For all integers $n, D \geq 0$ and precision $\eps>0$, there exists a distribution ${\cal D}_{n,D,\eps}$ over matrices in $\R^{O(\log(n)/\eps^2) \times D}$ such that the following holds: for any set of vectors $x_i \in \R^D$, $i \in [n]$, it holds that if $\Pi \sim {\cal D}_{n,D,\eps}$ then with high probability
\[
\| \Pi x_i \|^2
= (1\pm\eps) \| x_i \|^2, \quad \forall i \in [n].
\]
Moreover, the matrix $\Pi \sim {\cal D}_{n,D,\eps}$ can be sampled in time $\tO(D/\eps^2)$.
\end{lemma}

To use this lemma, we start by rewriting the leverage scores as vector norms:
\[
\sigma^{B}_i(A)
= a_i^T (B^T B)^+ a_i
= a_i^T (B^T B)^+ B^T B (B^T B)^+ a_i
= \Big\| B (B^T B)^+ a_i \Big\|^2.
\]
Now, by \cref{lem:JL}, we can sample a random matrix $\Pi \in \R^{O(\log(n)/\eps^2) \times D}$ (independent of $A,B$) that with high probability (in $n$) satisfies
\[
\Big\| \Pi B (B^T B)^+ a_i \Big\|^2
= (1 \pm \eps) \Big\| B (B^T B)^+ a_i \Big\|^2,
\quad \forall i \in [n].
\]
Having already computed $(B^TB)^+$, we can compute the matrix $C = \Pi B (B^T B)^+ \in \R^{O(\log(n)/\eps^2) \times d}$ in additional time $O(\frac{d \log(n)}{\eps^2} \cdot (D+d))$: we first compute the $O(d\log(n)/\eps^2)$ entries of $\Pi B$, each of which is an inner product in dimension $D$, and then we do the same to compute $\Pi B (B^T B)^+$, this time with inner products in dimension $d$. After this we can return estimates $\tilde\sigma_i = \| C a_i \|^2$ for the leverage scores, satisfying $\tilde\sigma_i = (1 \pm \eps) \sigma^{B}_i(A)$, at a cost per estimate of $1$ row-query to $A$ and time $O(r_A \log(n)/\eps^2)$.
This leads to the following algorithm.

\begin{algorithm}[H]
\caption{Generalized leverage scores via Johnson-Lindenstrauss, cf.~\cite[Lemma~7]{cohen2015uniform}} \label{alg:LS-via-JL}
\Input{matrices $A \in \R^{n \times d}$, $B \in \R^{D \times d}$, approximation factor $\eps>0$}

\BlankLine

\textbf{$-$ Preprocessing:}

sample random matrix $\Pi \sim {\cal D}_{n,D,\eps}$ and compute $C = \Pi B (B^T B)^+$\;
sample random matrix $\Pi' \sim {\cal D}_{n,d,\eps}$ and compute $C' = \Pi' (I - (B^T B) (B^T B)^+)$\;

\textbf{$-$ Leverage score computation (index $i \in [n]$):}

compute $C' a_i$; if $C' a_i \neq 0$ then return $\tilde\sigma_i = \infty$\;

otherwise, compute and return $\tilde\sigma_i = \| C a_i \|^2$\;

\end{algorithm}

Based on this algorithm we can prove the following lemma.

\begin{lemma}[Johnson-Lindenstrauss leverage scores] \label{lem:approx-LS}
Assume query access to $A \in \R^{n \times d}$, and assume we are explicitly given an $r_B$-sparse matrix $B \in \R^{D \times d}$ with $n,D \gg d$.
\cref{alg:LS-via-JL}, after a preprocessing of time $O(d^\omega + \min\{D d^{\omega-1}, D r_B^2\} + \frac{d D\log(n)}{\eps^2})$, returns with high probability estimates $\tilde \sigma_i$ for any $i \in [n]$ that satisfy $\tilde \sigma_i \approx_\eps \sigma_i^B(A)$ at a cost per estimate of one row query to $A$ and time $O(r_A \log(n)/\eps^2)$.
\end{lemma}
\begin{proof}
First we show correctness of the algorithm.
By \cref{lem:JL} we know that with high probability both $\| C a_i \| = (1 \pm \eps) \| B (B^T B)^+ a_i \|$ and $\| C' a_i \| = (1 \pm \eps) \| (I - (B^T B)(B^T B)^+) a_i \|$ for all $i \in [n]$.
From $\| C' a_i \|$ we can now check whether $a_i \perp \ker(B)$ or not: indeed $I - (B^T B) (B^T B)^+$ corresponds to the projector onto $\ker(B)$, and so $a_i \perp \ker(B)$ iff $\| C' a_i \| = 0$.
Hence, if $\| C' a_i \| = 0$ we can set $\tilde\sigma_i = \infty$ and otherwise we can set $\tilde\sigma_i = \| C a_i \|^2$, which will be correct with high probability.

Now we prove the complexity bounds.
First consider the precomputation phase.
As discussed earlier, given $B$ we can compute $(B^T B)^+$ and $(B^T B)(B^T B)^+$ in time $O(\min\{D d^{\omega-1}, D r_B^2 + d^\omega\})$, and $C = \Pi B (B^T B)^+$ and $C' = \Pi' (I - (B^T B) (B^T B)^+)$ in time $O(d D \log(n)/\eps^2)$.

For the approximation of a single leverage score $\sigma^B_i(A)$, we query the entire row $a_i$ of $A$, and compute $C' a_i$ and $C a_i$ in time $O(r_A \log(n)/\eps^2)$.
\end{proof}

This improves on direct computation as in \cref{lem:direct-LS} when $\eps$ is not too small.

\subsubsection{Grover search}

The previous section discussed how we can efficiently access the sampling weights, which are based on leverage scores.
Here we discuss how, given access to an appropriate weight vector $w \in (\R_{>0} \cup\{\infty\})^n$, we can sample a subsampled matrix $B \overset{w,\epsilon}{\leftarrow} A$.
For this we use a quantum sampling routine based on Grover search, as stated in the following lemma. 
For some intuition behind this key subroutine, note that if all $q_i$'s where exactly 0 or 1, then this problem is precisely that of finding $\sum_i q_i$ many marked elements in a list of $n$ elements, whose complexity is well-known to be $\Theta(\sqrt{n \sum_i q_i})$ using Grover search.
The extra work needed for the current setting amounts to first assuming query access to a string of random bits, based on which we mark for every $i$ whether we include $i \in S$, and then showing that query access to the random string can be simulated in sublinear time (see the later \cref{lem:random-oracle}).

\begin{lemma}[{\cite[Claim 3]{apers2022quantum}}] \label{lem:q-sampling}
Assume we have query access to a list $q = (q_i)_{i \in [n]} \in [0,1]^n$. 
There is a quantum algorithm that samples a subset $S \subseteq [n]$, such that $i \in S$ independently with probability $q_i$.
The algorithm uses $\tO(\sqrt{n \sum_i q_i})$ time and queries to $q$.
\end{lemma}

From this, we can prove the following lemma on a quantum algorithm for implementing a single iteration of the repeated halving algorithm.
By the repeated halving lemma (\cref{lem:repeated-halving}), with high probability, the returned matrix $B$ will have $O(d \log(d)/\eps^2)$ rows and satisfy $B \approx_\epsilon A$.

\begin{lemma}[Quantum halving algorithm] \label{lemma:single-iteration}
Assume query access to $A \in \R^{n \times d}$ with row sparsity $r$.
Consider a random submatrix $A' \subseteq_{1/2} A$ and assume that we are given $B' \in \R^{O(d \log d) \times d}$, with row sparsity $r$, such that $B'^T B' \approx_{1/2} A'^T A'$.
For $\epsilon \in (0,1]$, there is a quantum algorithm that with high probability returns a matrix $B \overset{w,\epsilon}{\leftarrow} A$ for $w_i = \min\{1,\tilde w_i\}$ with $2\sigma^{B'}_i(A) \leq \tilde w_i \leq 4\sigma^{B'}_i(A)$.
In expectation, the algorithm takes $\tO(\sqrt{nd}/\epsilon)$ row queries to $A$ and $\tO(r \sqrt{nd}/\epsilon + d^\omega)$ time.
\end{lemma}
\begin{proof}
First we invoke \cref{lem:approx-LS} on $A$ and $B'$, with $D = O(d \log(d))$.
It states that, after a precomputation time of $\tO(d^\omega)$, we can query approximate leverage scores $\tilde\sigma_i$ satisfying $\tilde\sigma_i \approx_{1/2} \sigma_i^{B'}(A)$ for $i \in [n]$, by querying one row of $A$ and using time $O(r_A \log(n))$ per leverage score.
Now set $w_i = \min\{1,\tilde w_i\}$ with $\tilde w_i = 2 \tilde\sigma_i$, which satisfies $2\sigma^{B'}_i(A) \leq \tilde w_i \leq 4\sigma^{B'}_i(A)$.
We define the list $q = (q_i)_{i \in [n]}$ by setting $q_i = \min\{1,c w_i \log(d)/\epsilon^2\}$ (as in \cref{sec:rep-halving}), and so we can query $q_i$ at the same cost of querying~$\tilde\sigma_i$.

To sample $B \overset{w,\epsilon}{\leftarrow} A$, we apply the algorithm from \cref{lem:q-sampling} to the list $q$.
\cref{thm:uniform} implies that $\sum_i q_i \in O(d \log(d)/\eps^2)$, and so the complexity amounts to $\tO(\sqrt{n d}/\eps)$ queries to $q$.
Combined with the previous paragraph, this gives a total complexity of $\tO(r_A \sqrt{nd}/\eps + d^\omega)$ time and $\tO(\sqrt{nd}/\eps)$ row queries to $A$.
\end{proof}

\subsubsection{Random oracle access} \label{sec:random-oracle}

To obtain our final quantum algorithm we would like to repeatedly apply \cref{lemma:single-iteration} in the repeated halving algorithm.
This is plug-and-play, up to one detail, which is how to query the downsampled matrices $A_L \subseteq_{1/2} \dots \subseteq_{1/2} A_1 \subseteq_{1/2} A$.
Notice that these can be defined implicitly by drawing $L$ bitstrings $\{X^{(\ell)} \in \{0,1\}^n\}_{\ell \in [L]}$ with $X^{(\ell)}_i = 1$ independently with probability $1/2$.
The matrix $A_k$ then contains the $i$-th row of $A$ if and only if $\Pi_{\ell \leq k} X^{(\ell)}_i = 1$.
Hence, if we could efficiently query such a ``random oracle'' then we're done.

Now, if we only care about query complexity, then we can sample these bitstrings explicitly and use QRAM access to them.
However, if we want to have a sublinear runtime as well, then explicitly drawing these bitstrings would be too expensive (scaling as $\Omega(n)$).
We remedy this by invoking the following lemma, stating that we can efficiently \emph{simulate} access to a random oracle. 
\begin{lemma}[{Quantum random oracles~\cite[Claim 1]{apers2022quantum}}] \label{lem:random-oracle}
Consider any quantum algorithm with runtime $T$ that makes queries to a uniformly random string.
We can simulate this algorithm with a quantum algorithm with runtime $\tO(T)$ without random string, using an additional QRAM of $\tO(T)$ bits.
\end{lemma}

For some intuition, the lemma builds on two facts.
First is that a quantum algorithm making $T$ queries to a string cannot distinguish whether this string is uniformly random, or only $2T$-wise independent~\cite{polynomialmethod}.
Second is that there are fast (classical) algorithms for constructing in time $\tO(T)$ a data structure that allows querying a $2T$-wise independent string, requiring only $\tO(1)$ time per query~\cite{christiani2015independence}.

\subsubsection{Quantum repeated halving algorithm} \label{sec:q-rep-halving}

We state our final algorithm below.

\begin{algorithm}[H]
\caption{Quantum repeated halving algorithm:} \label{alg:quantum-halving}
\Input{query access to $A \in \R^{n \times d}$, approximation factor $\eps>0$}

\BlankLine

implicitly construct the chain $A_L \subseteq_{1/2} \dots \subseteq_{1/2} A_1 \subseteq_{1/2} A$ for $L = O(\log(n/d))$\;

use Grover search\footnotemark{} to explicitly learn $A_L$; let $B_L = A_L$\;

\For{$\ell = L-1,L-2,\dots,1$}{ 
    use \cref{lemma:single-iteration} on $A_{\ell+1} \subseteq_{1/2} A_\ell$ and $B_{\ell+1}$ to explicitly construct $B_\ell \overset{w,1/2}{\leftarrow} A_\ell$ with $w_i = \min\{1,\tilde w_i\}$ and $2\sigma^{B_{\ell+1}}_i(A_\ell) \leq \tilde w_i \leq 4\sigma^{B_{\ell+1}}_i(A_\ell)$\;
}

use \cref{lemma:single-iteration} on $A_1 \subseteq_{1/2} A$ and $B_1$ to explicitly construct $B \overset{w,\epsilon}{\leftarrow} A$ with $w_i = \min\{1,\tilde w_i\}$ and $2\sigma^{B_1}_i(A) \leq \tilde w_i \leq 4\sigma^{B_1}_i(A)$\;

\Return{$B$}
\end{algorithm}

\footnotetext{More specifically, apply \cref{lem:q-sampling} with $q_i = \Pi_{\ell \leq L} X^{(\ell)}_i$, where the $X^{(\ell)}_i$'s were defined in \cref{sec:random-oracle}.}

Based on this algorithm, we can prove the following theorem, which implies our main theorem on quantum spectral approximation (\cref{thm:quantum-approx}).

\begin{theorem}[Quantum repeated halving algorithm] \label{thm:q-rep-halving}
With high probability, the quantum repeated halving algorithm (\cref{alg:quantum-halving}) returns a matrix $B$ with $O(d \log(d)/\epsilon^2)$ rows that satisfies $B^T B \approx_\epsilon A^T A$.
The algorithm makes $\tO(\sqrt{nd}/\epsilon)$ row queries to $A$ and takes time $\tO(d^\omega + r \sqrt{nd}/\epsilon)$.
\end{theorem}
\begin{proof}
Correctness of the algorithm directly follows from correctness of the original repeated halving algorithm (\cref{lem:repeated-halving}). It remains to prove the query and runtime guarantees.
By the discussion in \cref{sec:random-oracle}, we can assume efficient access to a random oracle to query the $A_\ell$ matrices in Step 1.
We use \cref{lem:q-sampling} to implement Step 2: we learn the $\Theta(d)$ row indices of $A_L$ among $A$'s $n$ rows in time $\tO(\sqrt{n d})$, and we explicitly learn the corresponding rows by making $\Theta(d)$ row queries to $A$.
For Step 3.~and 4., we repeatedly apply \cref{lemma:single-iteration}, requiring a total time of $\tO(d^\omega + r_A \sqrt{nd}/\eps)$ and $\tO(\sqrt{nd}/\eps)$ row queries to $A$.
\end{proof}

Based on this we can also prove the theorem about approximating leverage scores, which we restate below for convenience.
\quantumls*
\begin{proof}
The quantum algorithm first computes $B \in \R^{d \log(d)/\eps^2 \times d}$ such that $B^T B \approx_{\eps/2} A^T A$, with high probability.
By \cref{thm:q-rep-halving} this takes $\tO(\sqrt{nd}/\eps)$ row queries to $A$ and $\tO(d^\omega + r_A \sqrt{nd}/\eps)$ time.
Now note that $B^T B \approx_{\eps/2} A^T A$ implies that $\sigma^B_i(A) \approx_{\eps/2} \sigma_i(A)$, and so we can invoke either \cref{lem:direct-LS} or \cref{lem:approx-LS} to approximate $\sigma_i(A)$, with $D \in \tO(d/\eps^2)$ and $r_B = r_A = r$.
From \cref{lem:direct-LS} we get an additional preprocessing cost of $\tO(d^\omega + \min\{d^\omega,d r^2\}/\eps^2)$ and a cost per estimate of 1 row query to $A$ and time $O(r^2)$.
From \cref{lem:approx-LS} we get an additional preprocessing cost of $\tO(d^\omega + \min\{d^\omega,d r^2\}/\eps^2 + d^2/\eps^4)$ and a cost per estimate of 1 row query to $A$ and time~$\tO(r/\eps^2)$.
\end{proof}

\section{Quantum algorithm for \texorpdfstring{$\ell_p$}{ell-p}-Lewis weights} \label{sec:Lewis}

In this section we describe our quantum algorithm for approximating Lewis weights.
It combines a recent classical algorithm by Apers, Gribling and Sidford~\cite{apers2024lewis} with the quantum algorithms for spectral approximation and leverage score computation from the previous section. The classical algorithm from \cite{apers2024lewis} is a variation on an earlier algorithm by Lee \cite{YinTatThesis}. To motivate the algorithms and explain their differences, 
let us first recall and state some definitions.
For finite $p>0$ the $\ell_p$-Lewis weights $w^{(p)}(A) \in \R^n_{>0}$ of a matrix $A \in \R^{n \times d}$ are defined as the unique vector $w \in \R^n_{>0}$ such that 
\begin{equation}
    w_i 
= \sigma_i(W^{\frac{1}{2}-\frac{1}{p}}A) = w_i^{1-\frac{2}{p}} a_i^T(A^T W^{1-\frac{2}{p}} A)^{+} a_i, \quad \text{for all } i \in [n], \label{eq:FP}
\end{equation}
where $W = \diag(w)$.
The $i$-th Lewis weight $w_i$ is hence defined implicitly as the $i$-th leverage score of the matrix $W^{\frac12-\frac1p}A$. 
We are interested in $\eps$-multiplicative estimates $v \approx_\eps w$, which is a strong notion of approximation. Weaker notions are based on satisfying the fixed-point equation~\eqref{eq:FP} approximately; the weaker notions will lead to the stronger notion, after postprocessing and a loss of accuracy, as we explain next. Here we distinguish one-sided and two-sided approximations. We say that a vector $w$ is a one-sided $\eps$-approximation of the Lewis weights if it satisfies
\[
(1-\eps) \sigma_i(W^{\frac{1}{2}-\frac{1}{p}}A) \leq w_i \quad \text{for all } i \in [n], \quad \|w\|_1 \leq (1+\eps)d.
\]
The constraint $\|w\|_1 \leq (1+\eps)d$ ensures that the $\ell_1$-norm of $w$ is close to that of the vector of leverage scores, which is equal to $d$. Similarly, we call a vector $w$ a two-sided $\eps$-approximation of the Lewis weights if it satisfies
\[
(1-\eps) \sigma_i(W^{\frac{1}{2}-\frac{1}{p}}A) \leq w_i \leq (1+\eps) \sigma_i(W^{\frac{1}{2}-\frac{1}{p}}A) \leq w_i \quad \text{for all } i \in [n].
\]

We now turn to the algorithm that we use, \cref{alg:low-precision-Lewis}. As mentioned before, the algorithm is a variation on an earlier algorithm by Lee~\cite{YinTatThesis}. In fact, the vector $u$ constructed in lines 1 through 5 is exactly the output in Lee's algorithm: it iteratively applies the operator from the fixed-point equation~\eqref{eq:FP} and outputs the average $u$ of the iterates. Based on convexity of the functions $\phi_i(w) = \log(\sigma_i(W^{\frac{1}{2}-\frac{1}{p}}A)/w_i)$ for $i \in [n]$, it can be shown that $u$ is a one-sided $O(\delta)$-approximation of the Lewis weights.\footnote{It is claimed in \cite[Theorem 5.3.1]{YinTatThesis} that the algorithm returns a two-sided approximation, but the proof only shows a one-sided approximation guarantee (which suffices for the applications in \cite{YinTatThesis}). We thank Aaron Sidford for pointing this out to us.} This one-sided notion suffices for applications to $\ell_p$-embedding and $\ell_p$-regression problems, but not for our application in an IPM. It is shown in \cite[Theorem~2]{apers2024lewis} that the additional post-processing step in line $6$ of \cref{alg:low-precision-Lewis} improves the notion of approximation: the resulting vector $v$ is a two-sided $\delta'=O(\delta pd)$-approximation of the Lewis weights. We finally use a known stability result from~\cite[Lemma 14]{Fazel:highprecisionLewisWeights} which shows that a two-sided $\delta'$-approximation is an $O(\delta' \sqrt{d}p^2)$-multiplicative approximation. Choosing $\delta \in O(\eps/(p^3 d^{3/2}))$ thus ensures that the output $v$ is an entrywise $\eps$-multiplicative estimate of the true Lewis weights~$w^{(p)}(A)$: we have $v \approx_\eps w^{(p)}(A)$. The following theorem summarizes this. In the remainder of this section we analyze the query and time complexity of this algorithm, when we use the subroutines from the previous section to approximate leverage scores. 

\begin{theorem}[{\cite[Theorem 3]{apers2024lewis}}] \label{thrm:low-precision-lewis}
    For any matrix $A \in \R^{n \times d}$, $0<\eps<1$ and $p \geq 2$, \cref{alg:low-precision-Lewis} outputs $\eps$-multiplicative $\ell_p$-Lewis weights $v$ satisfying $(1-\eps) w^{(p)}_i(A) \leq v_i \leq (1+\eps) w^{(p)}_i(A)$, for all $i \in [n]$.
\end{theorem}

\begin{algorithm}[ht]
\caption{Approximate $\ell_p$-Lewis weights~\cite{apers2024lewis}}\label{alg:low-precision-Lewis}
\Input{$A \in \R^{n \times d}$, accuracy $0<\eps<1$, $p \geq 2$}
\Output{A vector $v \in \R^n$}

\BlankLine
Let $v_i^{(1)} = d/n$ for all $i \in [n]$, $\delta \in O(\eps/(p^3 d^{3/2}))$ and set $T = \lceil 2\log(n/d)/\delta \rceil$;

\For{$k=1,\ldots,T$}{ 
    Let $v^{(k+1)} \approx_{\delta} \sigma((V^{(k)})^{\frac12 - \frac1p} A)$\;
}

Let $u  = \frac{1}{T} \sum_{k=1}^T v^{(k)}$ and $s \approx_{\delta} \sigma(U^{\frac12 - \frac1p} A)$\;

\Return{$v$ with $v_i  = u_i (s_i/u_i)^{p/2} $}
\end{algorithm}

\subsection{Iterative quantum algorithm for approximate Lewis weights}

We obtain a quantum algorithm for computing Lewis weights by plugging our quantum algorithm for computing leverage scores (\cref{thm:q-rep-halving}) into \cref{alg:low-precision-Lewis}.
We have to be careful because in sublinear time we cannot afford to explicitly write down $\Omega(n)$ coefficients, and so we can only \emph{implicitly} keep track of the coefficients throughout the algorithm.

\begin{restatable}[Quantum Lewis weight approximation]{theorem}{quantumlw}
\label{thm:quantum-lw} \label{thm:quantum-Lewis}
Consider query access to a matrix $A \in \R^{n \times d}$ with row sparsity $r$ and $n \gg d$.
For any $0 < \eps \leq 1$ and $p \geq 2$, there is a quantum algorithm that, with high probability, provides query access to a vector $v \in \R^n$ satisfying $v \approx_\eps w^{(p)}(A)$.
The preprocessing cost of the algorithm is $\tO(p^6 d^{7/2} \sqrt{n}/\eps^2)$ row queries to $A$ and time
\[
\tO\left( \frac{p^9 d^{9/2}}{\eps^3} \left( \min\{d^\omega,d r^2\} + \sqrt{nd} r^2 \right) \right).
\]
After preprocessing, the cost per query to $v$ is 1 row query to $A$ and $O(p^3 d^{3/2} r^2/\eps)$ time.
\end{restatable}
For the reader's convenience, in terms of the parameter $\delta \in O(\eps/(p^3 d^{3/2}))$ defined in \cref{alg:low-precision-Lewis}, the cost of the algorithm can be seen as roughly $1/\delta$ times the cost of \cref{thm:quantum-ls} run with precision $\delta$. That is, the preprocessing cost  is $\tO(\sqrt{nd}/\delta^2)$ row queries to $A$ and time $\tO\left( \frac{1}{\delta} d^\omega + \frac{1}{\delta^3} \left( \min\{d^\omega,d r^2\} + r^2 \sqrt{nd} \right) \right)$. After preprocessing, the cost per query to $v$ is 1 row query to $A$ and $O(r^2/\delta)$ time. 

\begin{proof}
We establish the theorem by invoking \cref{thm:quantum-ls} inside \cref{alg:low-precision-Lewis}.
The theorem bounds quantities $p_q$, $p_t$ and $Q_t$ so that we can simulate query access to
\[
v^{(k+1)} \approx_\delta \sigma((V^{(k)})^{\frac12-\frac1p} A),
\qquad \text{and} \qquad
s \approx_{\delta} \sigma(U^{\frac12 - \frac1p} A)
\]
with the following costs:\footnote{We use that a row query to the $i$-th row of $(V^{(k)})^{\frac12-\frac1p} A$ can be simulated by 1 row query to the $i$-th row of $A$ and 1 query to~$v_i^{(k)}$.}
\begin{enumerate}
\item
\emph{Preprocessing:} $p_q$ row queries to $A$ and $v^{(k)}$ or $u$, and $p_t$ time.
\item
\emph{Cost per query to $v^{(k+1)}_i$:} 1 row query to $i$-th row of $A$ and $v^{(k)}$ or $u$, and $Q_t$ time.
    \end{enumerate} 
Specifically, the first item of \cref{thm:quantum-ls} gives bounds $p_q \in \tO(\sqrt{nd}/\delta)$, $p_t \in \tO(r \sqrt{nd}/\delta + d^\omega + \min\{d^\omega,d r^2\}/\delta^2)$ and $Q_t \in O(r^2)$.
It will follow from our argument that we can query $u = \frac{1}{T} \sum_{k=1}^T v^{(k)}$ in the same complexity as querying $v^{(k)}$ and $s$.
As a consequence, we can also query the output $v_i = u_i (s_i/u_i)^{p/2}$ in this complexity.

Bounding the cost of implementing \cref{alg:low-precision-Lewis} requires some care.
For starters, we bound the query and time complexity of a single query to $v_i^{(k)}$ or $s_i$ \emph{after preprocessing}.
A query to $v_i^{(k)}$ can be simulated with 1 row query to the $i$-th row of $A$, 1 query to $v_i^{(k-1)}$, and time $Q_t$.
On its turn, $v_i^{(k-1)}$ can be queried with 1 row query to the $i$-th row of $A$ (which we already queried!), 1 query to $v_i^{(k-2)}$, and time $Q_t$.
Repeating this argument, and noting that $v_i^{(1)} = n/d$, it follows that $v_i^{(k)}$ can be queried with 1 row query to the $i$-th row of $A$ and total time $O(k Q_t)$.
To query $u_i = \frac{1}{T} \sum_{k=1}^T v^{(k)}_i$, we note that it is not more expensive then the cost to query $v^{(T)}_i$, since the latter effectively queries each of $v^{(1)},\ldots,v^{(T-1)}$. Hence we can also query $u_i$ (and as a consequence also $s_i$) using $1$ row-query to~$A$ and time $O(T Q_t)$.

Now we bound the preprocessing cost for each step.
We solve the iterations chronologically, so that in iteration $k+1$ we can assume that the preprocessing is done for previous iterations.
That said, in iteration $k+1$ we have to make $p_q$ row queries to $A$ and $v^{(k)}$ (or $u$ in the final iteration), and spend time $p_t$.
By the preceding argument, we can make $p_q$ queries to $v^{(k)}$ by making $p_q$ row queries to $A$ and spending time $O(p_q k Q_t)$, so that in total we make $2 p_q$ row queries to $A$ and spend time $O(p_t + p_q k Q_t)$.
Summing over all $k=1,2,\dots,T$ iterations (plus the final postprocessing step 5.), we get a total row query complexity of $O(T p_q)$ and time complexity of $O(T p_t + T^2 p_q Q_t)$.

Recalling that $T \in \tO(1/\delta)$, and filling in the bounds on $p_q$, $p_t$ and $Q_t$ from \cref{thm:quantum-ls}, we proved the claimed complexities.
\end{proof}

\section{Approximate matrix-vector product}

In this section we describe our quantum algorithm for approximation of matrix-vector products.
It combines the quantum algorithm for spectral approximation (\cref{thm:quantum-approx}) with quantum multivariate mean estimation \cite{CHJ:multivariate}.
A blueprint of the algorithm is given in \cref{alg:quantumMv}.
As explained in the introduction, for our Newton step we wish to approximate a matrix-vector product $y = B v$ in the ``local inverse norm'', i.e., we wish to return $\tilde y$ such that $\| \tilde y - y \|_{W^{-1}} \leq \delta$ with $W = B^T B$.
Since the algorithm from \cite{CHJ:multivariate} only gives an approximation in the 2-norm, we first use a spectral approximation $\widetilde W \approx W$ to ``precondition'' the product to $z = \wt W^{-1/2} y$.
A 2-norm approximation of~$z$ then translates to a local norm approximation of $y$:
\[
\| \tilde z - z \|_2
= \| \widetilde W^{1/2} \tilde z - y \|_{\tilde W^{-1}}
\approx \| \widetilde W^{1/2} \tilde z - y \|_{W^{-1}},
\]
so that we can return $\tilde y = \widetilde W^{1/2} \tilde z$.

\begin{algorithm}[ht]
    \caption{Quantum algorithm to approximate matrix-vector products} \label{alg:quantumMv}
    \Input{$B \in \R^{n \times d}$, $v \in \R^{n}$, bound $\alpha \geq \|v\|_\infty$, accuracy $0<\delta<1$}
    \Output{A vector $\tilde y \in \R^d$}

    Compute $\widetilde W \approx_{1/2} B^T B$ using quantum spectral approximation (\cref{thm:quantum-approx})\;
    Compute $\widetilde W^{-1/2}$ classically\;
    Define the random variable $X
= - n v_\ell \widetilde W^{-1/2} B^T |\ell\rangle$, with $\ell \in [n]$ uniformly at random\;
    Use quantum multivariate mean estimation on $X$ with $T = \tO(\sqrt{n} \alpha d/\delta)$ quantum samples to obtain $\tilde \mu \in \R^d$ such that $\|\tilde\mu - \E[X]\|_2 \leq \delta$\;
    \Return{$\tilde y = \widetilde W^{1/2} \tilde \mu$}
\end{algorithm}

\approxmv
\begin{proof}
Let $W = B^T B$ and $y = B^T v$. Let $\gamma = \frac{1}{2}$. 
First, we use \cref{thm:quantum-approx} to construct a spectral approximation $\widetilde W \approx_\gamma W$ with $\tO(r \sqrt{nd})$ quantum queries and time $\tO(r\sqrt{nd} + d^\omega)$. We then classically compute $\widetilde W^{-1/2}$, which takes time $\tO(d^\omega)$.

Now, define the $d$-dimensional random variable
\[
X
= - n v_\ell \widetilde W^{-1/2} B^T |\ell\rangle,
\]
with $\ell \in [n]$ uniformly at random.
Then $X$ has mean
\[
\mu
= \E[X]
= - \frac{1}{n} \sum_{\ell = 1}^n n v_{\ell} \widetilde W^{-1/2} B^T |\ell\rangle
= \widetilde W^{-1/2} B^T v.
\]
To upper bound the trace of the covariance matrix $\Sigma$, note that
\begin{align*}
\sigma^2_i
= \E[X_i^2] - \E[X_i]^2
\leq \E[X_i^2]
&= n \braket{i|\widetilde W^{-1/2} B^T V^2 B \widetilde W^{-1/2}|i} \\
&\leq n \|v\|_\infty^2 \braket{i|\widetilde W^{-1/2} W \widetilde W^{-1/2}|i}
\leq (1+\gamma) n \|v\|_\infty^2,
\end{align*}
where we used that $\tilde W^{-1/2} W \tilde W^{-1/2} \approx_\gamma I$.
This implies that $\Tr(\Sigma) =\sum_{i=1}^d \sigma_i^2 \leq (1+\gamma) \|v\|_\infty^2 nd$. 

Now we can use quantum multivariate mean estimation~\cite[Thrm.~3.5]{CHJ:multivariate}.
We can get an estimate $\tilde\mu$ satisfying
\[
\| \tilde\mu - \mu \|_2
\leq \delta/2
\]
using $T = \tO(\sqrt{d\Tr(\Sigma)}/\delta) \leq \tO(\sqrt{\|v\|_\infty^2n}d/\delta)$ quantum samples of $X$.
Finally, we return the estimate $\tilde y = \widetilde W^{1/2} \tilde\mu$.
This estimate satisfies
\[
\| \tilde y - y \|_{W^{-1}}
\leq (1+\gamma) \| \tilde y - y \|_{\tilde W^{-1}}
= (1+\gamma) \| \tilde\mu - \mu \|_2
\leq (1+\gamma) \delta/2
\leq \delta.
\]
The total sample complexity is $T \in \tO(\sqrt{n}d\|v\|_\infty/\delta)$ and we can obtain a sample from $X$ using one row query to $B$, one query to $v$, and time $d r$ (recall that we already computed $\widetilde W^{-1/2}$).
The total time complexity is hence $\tO(r\sqrt{n}d^2 \|v\|_\infty/\delta)$.\footnote{Ref.~\cite{CHJ:multivariate} does not explicitly state the time complexity of their algorithm, but their overhead is essentially due to the quantum Fourier transform in dimension $d$ and, per quantum sample, the computation of $d$-dimensional inner products. This overhead is less than the time $dr$ that we spend to prepare a quantum sample.}
\end{proof}

\begin{remark} \label{rem: univariate}
While we omit an explicit treatment, one can also use \emph{univariate} quantum mean estimation~\cite{montanaro2015quantum} to estimate $\mu$ coordinate-wise.
While this improves the time complexity by a factor $\sqrt{d}$, it increases the number of queries by a factor $\sqrt{d}$.
\end{remark}

\section{Robustness of IPMs with respect to approximating the Newton step} 
\label{sec:robustness}

As we will see in \cref{sec: complexity analysis}, the techniques developed in the previous sections can be used to approximate a Newton step in the following way: we find (1) spectral approximations of Hessians and (2) approximations of gradients in the inverse-local norm. In this section we formally state the required precision with which we need to perform both operations in order to implement a variation of the ``short-step barrier method'' (cf.~\cite{Renegar88}). The argument is relatively standard; approximating the gradient in the inverse-local norm is equivalent to approximating the Newton step in the local norm (and for this it is well-known that a small constant error suffices), the second type of approximation leads to a quasi-Newton method and it can be analyzed as such. For completeness we include an elementary proof following Renegar's exposition~\cite{Renegar88} in \cref{sec: proofs robust update}. A similar variation was recently used to obtain a space-efficient interior point method in \cite{Space2023}.

We defer the definitions of \emph{self-concordance} and \emph{complexity} of a barrier to \cref{sec:SCB}, for now we just mention that all barrier functions that we consider are self-concordant and have known complexity parameters.

\begin{theorem}[IPM based on approximate Newton steps] \label{thm:IPM-master}
Consider the convex optimization problem $\min_{x \in D} c^T x$ for convex region $D$, and let $f:D \to \R$ be a self-concordant barrier function for $D$ with complexity $\vartheta_f$, gradient $g(x)$ and Hessian $H(x)$.
Assume that we are given an appropriate initial point $z_0 \in D$ (see \cref{remark:initial-point}).
Then, for any $\eps>0$, we give an IPM algorithm that returns $x' \in D$ satisfying $c^T x' \leq \mathrm{val} + \eps$.
The algorithm computes $O(\sqrt{\vartheta_f} \log(1/\eps))$ many approximate gradients $\tilde g(x)$ and Hessians $\tilde H(x)$ satisfying
\[
\| \tilde g(x) - g(x) \|_{H(x)} \leq 1/\polylog(n)
\quad \text{ and } \quad
\tilde H(x) \preceq H(x) \preceq \polylog(n) \, \tilde H(x).
\]
\end{theorem}

We will instantiate the above for linear programs with the logarithmic barrier, the volumetric barrier and the Lewis weight barrier.
At a first read, readers can treat this theorem as a black-box and proceed with \cref{sec: complexity analysis} for quantum implementations of such approximations. 

\begin{remark}[Initial point] \label{remark:initial-point}
As is customary for interior point methods, we assume that we are given some initial point $x_0$ in the interior of the feasible region $D$.
For \cref{thm:IPM-master} we assume that there exists a small but constant $\eta>0$ such that $z_0$ is an approximate minimizer of $f_\eta(x) = \eta c^T x + f(x)$.
More specifically, we require that the Newton step at $z_0$ is small in the local norm (see later for definitions). By a standard argument of ``path-switching'' (see, e.g., \cite[Sec.~2.4.2]{Renegar01}) this assumption is equivalent to the assumption that we are given an initial point $x_0$ in the interior of $D$ that is ``not too close to the boundary'' (e.g., it suffices to have $B(x_0,1/\poly(n,d)) \subseteq D \subseteq B(x_0,\poly(n,d))$ in the Euclidean norm).

For our statements about time complexity, we assume for simplicity that a query to an entry of $A$, $b$, or $c$ takes time $\widetilde O(1)$ and that all floating point calculations are done exactly.
While these are standard assumptions (see e.g.~\cite{cohen2021solving}), there do exist exceptions \cite{allamigeon2022interior}.
Generally, however, it holds that all computations can be carried out with a number of bits of accuracy that scales with the bit-complexity of $A$, $b$ and $c$ -- see for example~\cite{Renegar88}.
\end{remark}

\section{Quantumly approximating Hessians and gradients} \label{sec: complexity analysis}

From \cref{thm:IPM-master} it follows that we can implement an IPM by implementing approximate Newton steps, which can be derived from approximations of the gradient and Hessian of the barrier function.
Specifically, from a feasible point $x \in \R^d$ we need to compute approximations $Q(x)$ and $g(x)$ satisfying
\[
Q(x) \preceq H(x) \preceq C Q(x) \quad \text{ and } \quad
\| \tilde g(x) - g(x) \|_{H(x)^{-1}} \leq \zeta.
\]
for parameters $C$ and $\zeta$ satisfying $C,\zeta^{-1} \in \tO(1)$.
We first describe the high-level strategy for computing $Q(x)$ and $\tilde g$ and then in the next three subsections we give a detailed analysis of the complexity for the logarithmic, volumetric and Lewis weight barriers.

At a high level, for each of the barriers we will use our quantum spectral approximation algorithm (\cref{thm:quantum-approx}) to obtain a $\tO(1)$-spectral approximation $Q(x)$ of the Hessian. This theorem shows how to obtain a spectral approximation of $B^T B$ using row-queries to $B$. Our Hessians all take the form $A^T S_x^{-1} W S_x^{-1} A$ where $S_x$ is the diagonal matrix containing the slacks and $W \succ 0$ is a diagonal weight matrix. For the log-barrier, the volumetric barrier and the Lewis-weight barrier, the matrix~$W$ is respectively the identity, the diagonal matrix of leverage scores of $S_x^{-1} A$, and the diagonal matrix of Lewis weights of $S_x^{-1} A$. For the latter two barriers we thus need to compute constant factor approximations of leverage scores and Lewis weights respectively. We do so using \cref{thm:quantum-ls} and \cref{thm:quantum-Lewis} respectively. 

To obtain a good estimate of the gradient, we use our quantum algorithm for approximate matrix-vector products (\cref{thm:approx-mv}).
To apply the theorem we note that $g(x)$ takes the form $-A^T S_x^{-1} W \mathbf 1$, with $\mathbf 1$ the all-ones vector and $W$ the same rescaling matrix as before.
We interpret this as a matrix-vector product between $B=\sqrt{W} S_x^{-1} A$ and $v=\sqrt{W} \mathbf 1$.
This means that for the volumetric barrier and the Lewis weight barrier we cannot afford to query entries of $v$ exactly, instead we will query multiplicative approximations of these entries. The next lemma analyzes the error that such a multiplicative approximation incurs.
\begin{lemma} \label{lem: mult error}
    Let $B \in \R^{n \times d}$, $D \in \R^{n \times n}$ diagonal with $D_{ii} \in [1-\eps,1+\eps]$, and $v \in \R^n$, then 
    \[
    \|B^T v - B^T D v\|_{(B^T B)^{-1}} \leq \eps \|v\|_2
    \]
\end{lemma}
\begin{proof}
    We bound the squared norm as follows:
    \begin{align*}
        \|B^T v - B^T D v\|_{(B^T B)^{-1}}^2
        &= \|B^T (D-I) v \|_{(B^T B)^{-1}}^2 \\
        &= v^T (D-I) B(B^T B)^{-1} B^T (D-I) v \\
        &\leq v^T (D-I)^2 v
        \leq \eps^2 \|v\|_2^2 
    \end{align*}
    where in the first inequality we use that $B(B^T B)^{-1} B^T$ is a projector.
\end{proof}
For both the volumetric and Lewis weight barrier we have $\|v\|_2 \leq \sqrt{d}$. To achieve a small constant error in the gradient, it thus suffices to use $(1 \pm O(1/\sqrt{d}))$-multiplicative approximations of the entries of $v$.

\subsection{Logarithmic barrier} \label{sec:log-barrier}

The logarithmic barrier (cf.~\cite{Renegar88}) for a polytope $Ax \geq b$, with $A \in \R^{n \times d}$ and $b \in \R^n$, is defined as
\[
F(x) \coloneqq - \sum_{i=1}^n \log(a_i^T x - b_i),
\]
which is well defined in the interior of the polytope.
It is a self-concordant barrier with complexity~\cite[Section 2.3.1]{Renegar01}
\[
\vartheta_F
\in O(n).
\]
Its gradient and Hessian are 
\[
g_F(x) \coloneqq \nabla F(x) =- A^T S_x^{-1} \mathbf 1
\quad \text{ and } \quad
H_F(x) \coloneqq \nabla^2 F(x) = A^T S_x^{-2} A.
\]
The following lemmas bound the cost of approximating these using our quantum algorithms.
\begin{lemma}[Hessian of logarithmic barrier] \label{lem:hessian log}
Given query access to $A \in \R^{n \times d}$, $x \in \R^d$ and $b \in \R^n$, and a parameter $\eps>0$, we give a quantum algorithm that with high probability computes a matrix $Q(x)$ such that $Q(x) \approx_\eps H_F(x)$.
The algorithm uses $\tO(\sqrt{nd}/\eps)$ row queries to $A$, $x$, and $b$ and time $\tO(\sqrt{nd}r/\eps+d^\omega)$. 
\end{lemma}
\begin{proof}
We apply \cref{thm:quantum-approx} to the matrix $S_x^{-1}A$ and observe that we can implement one row-query to $S_{x}^{-1} A$ in time $O(r)$. Indeed, it suffices to observe that we can use $O(1)$ row queries and $O(r)$ time to compute the slack of a single constraint.
\end{proof}

\begin{lemma}[Gradient of logarithmic barrier] \label{lem:gradient log}
Given query access to $A \in \R^{n \times d}$, $x \in \R^d$ and $b \in \R^n$, we give a quantum algorithm that with high probability computes a vector $\tilde g(x)$ such that $\|\tilde g(x) - g_F(x)\|_{H_F(x)^{-1}} \leq \zeta$ using $\tO(\sqrt{n}d/\zeta)$ row queries to $A$, $x$, and $b$ and time $\tO(\sqrt{n}d^2 r/\zeta+d^\omega)$. 
\end{lemma}
\begin{proof}
As in the proof of \cref{lem:hessian log}, we set $B =S_x^{-1}A$ and observe that we can implement one row query to $S_{x}^{-1} A$ in time $O(r)$. We then use \cref{thm:approx-mv} with this $B$ and $v = \mathbf 1$ to obtain a vector $\tilde y$ such that $\|\tilde y - g_F(x)\|_{H_F(x)^{-1}} = \|\tilde y - g_F(x)\|_{(B^T B)^{-1}}\leq \zeta$. 
\end{proof}

By \cref{thm:IPM-master} we can solve an LP, returning a feasible point $\eps$-close to optimal, while making $O(\sqrt{\vartheta_f} \log(1/\eps))$ many approximate gradient and Hessian queries.
Combining \cref{lem:hessian log} and \cref{lem:gradient log} (with $\eps,\zeta = 1/\polylog(n)$), and using the fact that $\vartheta_F \in O(n)$, this yields a quantum IPM based on the logarithmic barrier which makes $\tO(n d)$ row queries and uses time $\tO(n d^2 r + \sqrt{n} d^\omega)$.
While useful as a proof of concept, the query complexity is trivial, and the time complexity is worse than the best classical algorithm which runs in time  $\tO(nd + d^3)$ \cite{brand2020solving}.
We hence turn to a more complicated barrier.

\subsection{Volumetric barrier} \label{sec:volumetric-barrier}

The volumetric barrier for a polytope $Ax \geq b$, introduced by Vaidya~\cite{Vaidya96} (see also~\cite{VaidyaAtkinson93,Anstreicher97volumetricLP}), is defined as
\[
V(x) \coloneqq \frac{1}{2} \ln \det(H_F(x)) = \frac{1}{2} \ln \det(A^T S_x^{-2} A).
\]
The volumetric barrier is a self-concordant barrier with complexity 
\[
\vartheta_V
\in O(\sqrt{n}d).
\]
An appropriately weighted linear combination of the volumetric barrier and the logarithmic barrier, called the \emph{hybrid barrier}, has complexity parameter $O(\sqrt{nd})$. For ease of notation, we focus here on approximating the Hessian and gradient of the volumetric barrier.\footnote{The hybrid barrier is $V_\rho(x) = V(x) + \rho F(x)$ where $\rho=(d-1)/(n-1)$, following the notation of \cite{Anstreicher97volumetricLP}. The analogue of~\cref{eq:volumetric-approximation} thus holds for the matrix $\widetilde H_{V_\rho}(x) = A^T S_x^{-1}(\Sigma_x + \rho I) S_x^{-1} A$, and the gradient is $g_{V_\rho}(x) = A^T S_x^{-1}(\Sigma_x + \rho I)\mathbf 1$. To obtain \cref{lem:hessian vol,lem:gradient vol} for the hybrid barrier it thus suffices to replace $\Sigma_x$ by $\Sigma_x + \rho I$ in the proofs (and use the fact that $a \approx_\eps b$ implies $a+\rho \approx_\eps b+\rho$ for $a,b,\rho>0$). \label{footnote:hybrid}}

Given a feasible $x \in \R^d$, let $\sigma_x$ be the vector of leverage scores of $S_x^{-1} A$ and let $\Sigma_x = \Diag(\sigma_x)$.
Then the gradient of $V$ at $x$ has the simple form
\[
g_V(x)
\coloneqq \nabla V(x)
= - A^T S_x^{-1} \Sigma_x \mathbf 1.
\]
The Hessian of $V$ has a more complicated form, but it admits a constant factor spectral approximation that has a simple form. 
We have 
\begin{equation} \label{eq:volumetric-approximation}
\wt H_V(x) \preceq H_V(x) \preceq 5 \wt H_V(x),
\end{equation}
where 
\[
\wt H_V(x) = A^T S_x^{-1} \Sigma_x S_x^{-1} A. 
\]
We can efficiently approximate this approximate Hessian.

\begin{lemma}[Hessian of volumetric barrier] \label{lem:hessian vol}
Given query access to $A \in \R^{n \times d}$, $x \in \R^d$ and $b \in \R^n$, we give a quantum algorithm that computes with high probability a matrix $Q(x)$ such that $Q(x) \approx_\eps \wt H_V(x)$.
The algorithm uses $\tO(\sqrt{nd}/\eps)$ row queries to $A$, $x$, and $b$ and time $\tO(\sqrt{nd}r^2/\eps+d^\omega/\eps^2)$.
\end{lemma}

\begin{proof}
We use \cref{thm:quantum-ls} to set up query access to $\tilde \sigma_x$, a vector of $(1 \pm \eps/3)$-multiplicative approximations to the leverage scores of $S_{x}^{-1} A$. This requires a preprocessing that uses $\tO(\sqrt{nd}/\eps)$ row queries to $S_x^{-1}A$ and time\footnote{For convenience we replace the term $\min\{d^\omega,d r^2\}$ by $d^\omega$.} 
$\tO(\sqrt{nd} r/\eps + d^{\omega}/\eps^2)$; afterwards we can query an entry of $\tilde \sigma_x$ using one row query to $S_x^{-1} A$ and time $\tO(r^2)$.
Let $\wt\Sigma_x = \diag(\tilde\sigma_x)$.

We then apply \cref{thm:quantum-approx}, with approximation factor $\eps/3$, to the matrix $\wt \Sigma_x^{1/2}S_x^{-1}A$ and observe that we can implement one row query to $\widetilde \Sigma_x^{1/2} S_{x}^{-1} A$  using one query to $\widetilde \sigma_x$ and one row query to $S_{x}^{-1}A$. The cost of row queries to the latter follows from the observation that we can use $O(1)$ row queries to $A,b$, and time $O(r)$, to compute the slack of a single constraint.
Let $B$ be the matrix returned by the algorithm from \cref{thm:quantum-approx}.
With high probability, $B$ has $\tO(d)$ rows and satisfies 
\[
B^T B \approx_{\eps/3} A^T S_x^{-1} \widetilde \Sigma_x S_x^{-1} A.
\]
Since we also have 
\[
A^T S_x^{-1} \widetilde \Sigma_x S_x^{-1} A
\approx_{\eps/3} \wt H_V(x),
\]
this shows that 
\[
B^T B \approx_\eps  \wt H_V(x),
\]
where we used that $1-\eps \leq (1-\eps/3)^2$ and $(1+\eps/3)^2 \leq 1+\eps$.
\end{proof}

\begin{lemma}[Gradient of volumetric barrier] \label{lem:gradient vol}
Given query access to $A \in \R^{n \times d}$, $x \in \R^d$ and $b \in \R^n$, we give a quantum algorithm that with high probability computes a vector $\tilde g(x)$ such that $\|\tilde g(x)-g(x)\|_{H_V(x)^{-1}} \leq \zeta$ using $\tO(\sqrt{n}d/\zeta)$ row queries to $A$, $x$ and $b$, and time $\tO(\sqrt{n}d^2 r/\zeta + d^{\omega+1}/\zeta^2)$. 
\end{lemma}
\begin{proof}
Let $\delta>0$ be a parameter that we choose later. We first use \cref{thm:quantum-ls} to set up query access to $\tilde \sigma_x$, a vector of $(1\pm \delta)$-multiplicative approximations to the leverage scores of $S_x^{-1}A$. We then let $\widetilde B = \widetilde \Sigma_x^{1/2} S_x^{-1} A$ and $\tilde v = \widetilde \Sigma_x^{1/2} \mathbf 1$ and use \cref{thm:approx-mv} to obtain a vector $\tilde y \in \R^d$ such that $\|\tilde y - \widetilde B^T \tilde v \|_{(\widetilde B^T \widetilde B)^{-1}} \leq \zeta$.
By \cref{eq:volumetric-approximation} and a triangle inequality we can bound
\begin{align*}
\|\tilde y - g(x)\|_{H_V(x)^{-1}} &\leq \|\tilde y - g(x)\|_{(\widetilde B^T \widetilde B)^{-1}} \leq \|\tilde y - \widetilde B^T \tilde v\|_{(\widetilde B ^T \widetilde B)^{-1}}  + \| g(x) - \widetilde B^T \tilde v\|_{(\widetilde B ^T \widetilde B)^{-1}}. 
\end{align*}
For the first term, we have the upper bound $\|\tilde y - \widetilde B^T \tilde v\|_{(\widetilde B ^T \widetilde B)^{-1}} \leq \zeta$ by construction. For the second term, we first rewrite $g(x)$ as 
\[
g(x) = A^T S_x^{-1} \Sigma_x \mathbf 1 = (A^T S_x^{-1} \widetilde \Sigma_x^{1/2} ) (\widetilde \Sigma_x^{-1/2} \Sigma_x \mathbf 1) = \widetilde B^T v'
\]
where $v' = \widetilde \Sigma_x^{-1/2} \Sigma_x \mathbf 1$. We then note that $v' = \Sigma_x \widetilde \Sigma_x^{-1}\tilde v$, i.e., entrywise $v'$ is a $(1\pm \delta)$-approximation of $\tilde v$. \cref{lem: mult error} thus shows that $\| g(x) - \widetilde B^T \tilde v\|_{(\widetilde B ^T \widetilde B)^{-1}} \leq \delta \|\tilde v\|_2 \leq \delta \sqrt{(1+\delta)d}$, where the last inequality uses that $\|\tilde v\|_2 = \sqrt{\sum_{i \in [n]} (\tilde \sigma_x)_i} \leq \sqrt{(1+\delta)d}$. Setting $\delta = \zeta/\sqrt{d}$ shows that $\|\tilde g - g(x)\|_{H_V(x)^{-1}} = O(\zeta)$. 

We finally establish the query and time complexity. We use \cref{thm:quantum-ls} to set up query access to~$\tilde \sigma_x$. Since $1/\delta^2>r$, it is advantageous to use the first (direct) approach that uses as preprocessing $\tO(\sqrt{n}d/\zeta)$ row queries to $S_x^{-1}A$ and time $\tO(d^{\omega+1}/\zeta^2 + r \sqrt{n}d/\zeta)$, and then allows us to query an entry of $\tilde \sigma_x$ at a cost of $1$ row query to $S_x^{-1} A$ and time $O(r^2)$.  The application of \cref{thm:approx-mv} uses $\tO(\sqrt{n} d/\zeta)$ row queries to $\widetilde \Sigma_x^{1/2} S_x^{-1} A$ and time $\tO(r\sqrt{n}d^2/\zeta + d^\omega)$ (here we are using that the time needed per query to $\widetilde \sigma_x$ is $r^2$ and thus less than $dr$).
\end{proof}

By \cref{thm:IPM-master} we can solve an LP to precision $\eps$ while making $O(\sqrt{\vartheta_f} \log(1/\eps))$ many approximate gradient and Hessian queries.
Combining \cref{lem:hessian vol} and \cref{lem:gradient vol} (with $\eps,\zeta = 1/\polylog(n)$), and using the fact that $\vartheta_V \in O(\sqrt{nd})$, this yields a quantum IPM based on the volumetric barrier which makes $\tO(n^{3/4} d^{5/4})$ row queries and uses time
\[
\tO(n^{1/4} d^{1/4} ( \sqrt{n} d^2 r + d^{\omega+1})),
\]
which is $n^{3/4} \poly(d,\log(n),1/\eps)$.\footnote{Strictly speaking, this complexity uses the hybrid barrier, see \cref{footnote:hybrid}.}
For $n$ sufficiently large the query and time complexity are sublinear and beat that of classical algorithms.
However, the $n^{3/4}$-scaling is worse than the $\sqrt{n}$-scaling of quantum algorithms based on the cutting plane method (\cref{app: cutting plane}).
Hence we turn to the final barrier.

\subsection{Lewis weight barrier} \label{sec:LW-barrier}

Finally, we discuss the Lewis weight barrier, as introduced by Lee and Sidford in \cite{LeeSidford19}.
For integer $p > 0$, it is defined as 
\[
\psi(x) = \ln \det(A^T S_x^{-1} W_x^{1-\frac{2}{p}} S_x^{-1} A) \quad \text{where } W_x = \Diag(w^{(p)}(S_x^{-1}A)).
\]
Its gradient is 
\[
g_\psi(x)
\coloneqq \nabla \psi(x) = - A^T S_x^{-1} W_x \mathbf 1.
\]
While the exact Hessian might again take a complicated form, we can approximate it by
\[
\wt H_\psi(x)
\coloneqq A^T S_x^{-1} W_x S_x^{-1} A. 
\]
Indeed, it holds that
\[
\wt H_\psi(x) \preceq \nabla^2 \psi(x) \preceq (1+p) \wt H_\psi(x).
\]

Letting $v_p = (p+2)^{3/2} n^{\frac{1}{p+2}} + 4 \max\{p,2\}^{2.5}$, Lee and Sidford showed \cite[Theorem 30]{LeeSidford19} that the barrier function $v_p^2 \cdot \psi(x)$ is a self-concordant barrier function with complexity $\vartheta_\psi \in O(d v_p^2)$.
We will be considering $p = \polylog(n)$, in which case the barrier has complexity
\[
\vartheta_\psi
\in \tO(d).
\]
For the complexity statements that follow, we assume $p = \polylog(n)$ so that polynomial factors in $p$ are absorbed in the $\tO$-notation.

To approximate the Hessian and gradient for the Lewis weight barrier, we follow the same approach as for the volumetric barrier, but with leverage scores replaced by Lewis weights.
The only part of the proofs that changes is the complexity analysis. 
\begin{lemma}[Hessian of the Lewis weight barrier] \label{lem:hessian lew}
Given query access to $A \in \R^{n \times d}$, $x \in \R^d$ and $b \in \R^n$, we give a quantum algorithm that with high probability computes a matrix $Q(x)$ such that $Q(x) \approx_\eps \wt H_\psi(x)$.
The algorithm uses $\tO(\sqrt{n}d^{7/2}/\eps^2)$ row queries to $A$, $x$, and $b$ and time $\tO\left(\frac{d^{9/2}}{\eps^3}\left(r^2\sqrt{nd}+d^\omega\right)\right)$.
\end{lemma}
\begin{proof}
    We follow the same proof strategy as the one we used for \cref{lem:hessian vol}.
    We refer to that proof for correctness of the choices of parameters that follow.
    We use \cref{thm:quantum-Lewis} to set up query access to $\tilde w_x$, a vector of $(1 \pm \eps/3)$-multiplicative approximations to the $\ell_p$-Lewis weights of $S_x^{-1} A$.
    This requires a preprocessing of $\tO(\sqrt{n} d^{7/2}/\eps^2)$ row queries to $S_x^{-1} A$ and time $\tO(\frac{d^{9/2}}{\eps^3} \cdot \left(r^2 \sqrt{nd} + d^{\omega}\right))$.
    Each query to $\tilde w_x$ then requires~$1$ row-query to $S_x^{-1} A$ and time $\tO(r^2 d^{3/2}/\eps)$.
    We then apply \cref{thm:quantum-approx}, with desired approximation factor $\eps/3$, to the matrix $\wt W_x^{1/2}S_x^{-1}A$ and observe that we can implement one row query to $\wt W_x^{1/2} S_{x}^{-1} A$ using one query to $\widetilde w_x$ and one row query to $S_{x}^{-1}A$.
    The cost of row queries to the latter follows from the observation that we can use $O(1)$ row queries and time $O(r)$ to compute the slack of a single constraint. Observe that the query and time complexity of setting up query access to $\tilde w_x$ dominate the complexity of the algorithm.
\end{proof}

\begin{lemma}[Gradient of the Lewis weight barrier] \label{lem:gradient lew}
Given query access to $A \in \R^{n \times d}$, $x \in \R^d$ and $b \in \R^n$, we give a quantum algorithm that with high probability computes a vector $\tilde g(x) \in \R^d$ such that $\|\tilde g(x)-g(x)\|_{H(x)^{-1}} \leq \zeta$ using $\tO(\sqrt{n}d^{9/2}/\zeta^2)$ row queries to $A$, $x$, and~$b$ and time $\tO\left(\frac{d^{6}}{\zeta^3}\left(r^2\sqrt{nd}+d^\omega\right)\right)$. 
\end{lemma}
\begin{proof}
        We follow the same proof strategy as the one we used for \cref{lem:gradient vol}.
        We refer to that proof for correctness of the choices of parameters that follow.
        The dominant term in both the query and time complexity is the use of \cref{thm:quantum-Lewis} to set up query access to $\tilde w_x$, a vector of $\zeta/\sqrt{d}$-multiplicative approximations to the leverage scores of $S_x^{-1} A$.
        This requires a preprocessing of $\tO(\sqrt{n} d^{9/2}/\zeta^2)$ row queries to $S_x^{-1} A$ and time\footnote{For convenience we replace the term $\min\{d^{\omega},d r^2\}$ by $d^{\omega}$.} $\tO(\sqrt{n} d^{13/2} r^2/\zeta^3 + d^{\omega+6}/\zeta^3)$.
        Each query to $\tilde w_x$ then requires~$1$ row-query to $S_x^{-1} A$ and time $\tO(r^2d^2/\zeta)$.
        As before, we then apply \cref{thm:approx-mv} to $\widetilde W_x^{1/2} S_x^{-1} A$ and $\widetilde W_x^{1/2} \mathbf 1$ with desired precision $\zeta$. This uses $\tO(\sqrt{n} d/\zeta)$ row queries to $\widetilde W_x^{1/2} S_x^{-1} A$ and time $\tO((r^2 d/\zeta)\sqrt{n}d/\zeta + d^\omega)$.
        Note that unlike in the proof of \cref{lem:gradient vol}, here the time needed per query to $W_x$ is $r^2 d^2/\zeta$ and this exceeds the additional computational cost of $d r$ that the multivariate quantum mean estimation algorithm uses per quantum sample.
        However, as for the previous lemma, observe that the query and time complexity of setting up query access to $\tilde w_x$ dominates the complexity of the algorithm.
\end{proof}

Again invoking \cref{thm:IPM-master}, we can solve an LP to precision $\eps$ while making $O(\sqrt{\vartheta_f} \log(1/\eps))$ many approximate gradient and Hessian queries.
Combining \cref{lem:hessian lew} and \cref{lem:gradient lew} (with $\eps,\zeta = 1/\polylog(n)$), and using the fact that $\vartheta_V \in \tO(d)$, this yields a quantum IPM based on the Lewis weight barrier which makes $\tO(\sqrt{n} d^5)$ row queries and uses time
\[
\tO(\sqrt{d} ( \sqrt{n}d^{13/2}r^2+d^{\omega+6} )),
\]
which is $\sqrt{n} \poly(d,\log(n),1/\eps)$.
This proves our main \cref{thm:quantum-IPM}.
For $n$ sufficiently large the query and time complexity are sublinear and beat that of classical algorithms.
The $\sqrt{n}$-scaling matches that of a quantum algorithm based on the cutting plane method (\cref{app: cutting plane}), although the $d$-dependence is significantly worse.
See \cref{sec:main-discussion} for a discussion on this benchmark comparison, and for open directions on how to improve the $d$-dependency of our IPM.

\section{Lower bounds} \label{sec: LB}

Here we discuss some relevant lower bounds.
In \cref{sec:spectral-approx-LB} we show that the complexity of our algorithm for spectrally approximating $A^T A$ is tight in terms of the number of row queries to $A$: $\Omega(\sqrt{nd})$ quantum row queries are needed.
In \cref{sec:LP-LB} we reproduce a bound from~\cite[Cor.~30]{vAGGdW:quantumSDP} which shows that $\Omega(\sqrt{nd}r)$ queries to $A$ are needed to solve LPs up to constant additive error, even when the entries of $A$, $b$, and $c$ are restricted to $\{-1,0,1\}$. In both sections we moreover discuss corresponding lower bounds in two classical query models: the (usual) randomized query complexity, and a quantum-inspired model from~\cite{chia22sampling}. 

We briefly introduce the quantum-inspired query model here. For vectors, we say that we have \textit{sampling and query access} access to a vector $v \in \R^n$, denoted $\mathrm{SQ}(v)$, if we can efficiently perform the following three operations: 
(1) \textit{query}: given an index $i \in [n]$, output the entry $v_i$. (2) \textit{sample}: sample an index $j \in [n]$ with probability $|v_j|^2/\|v\|^2$. (3) \textit{norm}: output the $\ell_2$-norm $\|v\|$ of $v$. For matrices, we say that we have \textit{sampling and query access} to a matrix $A \in \R^{n \times d}$, denoted $\mathrm{SQ}(A)$, if we have sampling and query access $\mathrm{SQ}(a_i)$ for each row $a_i$ of $A$, and also sampling and query access $\mathrm{SQ}((\|a_i\|)_{i \in [n]})$ to the vector of row norms.

\subsection{Spectral approximation lower bound} \label{sec:spectral-approx-LB}

Here we discuss several complementary lower bounds for spectral approximation.
The first one applies to arbitrary $n \geq d$ but constant $\eps$.
It easily proven by a reduction to quantum search.

\begin{theorem}
    For arbitrary $n \geq d$, there is a family of matrices $A \in \R^{n \times d}$ for which finding a constant factor spectral approximation of $A^T A$ requires $\Omega(\sqrt{nd})$ quantum row queries to $A$, and $\Omega(n)$ classical row queries to $A$. 
\end{theorem}
\begin{proof}
    Without loss of generality, assume $d$ divides $n$. 
    Given $d$ bitstrings $z_1,\ldots, z_d \in \{0,1\}^{n/d}$ we construct a matrix $A$ that consists of $d$ blocks of size $n/d \times d$: the $j$th block contains $z_j$ in the $j$th column and is zero elsewhere. Note that a row query to $A$ can be simulated with one query to one of the $z_j$'s. Moreover, the matrix $A^T A$ has a particularly simple form: $A^T A = \diag(|z_1|,|z_2|,\ldots,|z_d|)$.
    The diagonal of a spectral $1/2$-approximation of $A^T A$ contains $1/2$-approximations of the diagonal of $A^T A$ and therefore allows us to compute the string $(\OR(z_1),\OR(z_2),\ldots,\OR(z_d))$. To conclude the proof, we observe that computing the string $(\OR(z_1),\OR(z_2),\ldots,\OR(z_d))$ is known to require $\Omega(d \sqrt{n/d}) = \Omega(\sqrt{nd})$ quantum queries to the bitstrings $z_1,\ldots,z_d$. This follows for instance from combining the $\Omega(\sqrt{n/d})$ quantum query lower bound on $\OR$ over $n/d$ bits \cite{boyer1998tight} with the strong direct product theorem for quantum query complexity \cite{lee2013strong}.
    The randomized query complexity of the same problem can similarly be shown to be $\Omega(d \cdot n/d) = \Omega(n)$, using the randomized query complexity $\Theta(n/d)$ of $\OR$ over $n/d$ bits and a direct product (or direct sum) theorem for randomized query complexity~\cite{jain2010sum}.
\end{proof}

The second lower bound also shows optimality of the $\eps$-dependency, but it is restricted to $n \in O(d^2)$.
The bound follows from a lower bound on graph sparsification in
\cite{apers2022quantum}.
There it is shown that any quantum algorithm for computing an $\eps$-spectral sparsifier of a graph $G$ must make $\Omega(\sqrt{nd}/\eps)$ queries to its adjacency list, where $d$ is the number of vertices of $G$, $n \in O(d^2)$ is the number of edges of $G$, and $\eps \in \Omega(\sqrt{d/n})$.
This lower bound applies to our setting if we let $A$ denote the edge-vertex incidence matrix of a graph $G$, in which case an adjacency list query can be reduced to a row query to $A$, and returning an $\eps$-spectral approximation of $A$ is equivalent to returning an $\eps$-spectral sparsifier of $G$.
We hence get the following lower bound.

\begin{lemma}[{\cite[Theorem 2]{apers2022quantum}}] \label{lemma:spectral-LB}
For arbitrary $d$, $n \in O(d^2)$ and $\eps \in \Omega(\sqrt{d/n})$, there is a family of matrices $A \in \R^{n \times d}$ for which finding an $\eps$-spectral approximation of $A^T A$ requires $\wt\Omega(\sqrt{nd}/\eps)$ quantum row queries to $A$.
\end{lemma}
In contrast, randomized classical algorithms exist that produce an $\eps$-spectral sparsifier of a graph $G$ in time $\tO(m)$, which is known to be optimal, see e.g.~the discussion in \cite{apers2022quantum} for more details. 

The third lower bound applies to the quantum-inspired sampling and query access model. 
\begin{lemma}
    For arbitrary $n \geq d \geq 2$, there is a family of matrices $A \in \R^{n \times d}$ for which finding a constant factor spectral approximation of $A^T A$ requires $\Omega(n)$ queries to $\mathrm{SQ}(A)$ and $\mathrm{SQ}(A^T)$.
\end{lemma}
\begin{proof}
    We consider the family of matrices that are constructed as follows: let $z \in \{0,1\}^n$ with $|z|=0$ or $|z|=1$. We let $A$ be the all-ones matrix of size $n \times d$, except for the last column whose $i$th entry is $(-1)^{z_i}$. The absolute value of each of the entries of $A$ is equal to $1$. We can thus implement sampling access to each of the rows (or columns) without queries to $z$. The same holds for the vector of row-norms (or column-norms). A query to an entry of a row can be implemented using at most one query to $z$. 

    The $\Omega(n)$-query lower bound then follows by observing that $A^TA$ has rank $1$ if $|z|=0$ and rank $2$ if $|z|=1$. A constant factor spectral approximation can be used to determine the rank of $A^T A$, and thus $\OR(z)$. 
\end{proof}

\subsection{LP solving lower bound} \label{sec:LP-LB}

To establish the lower bounds in the quantum and randomized query model, we will use a small variation of~\cite[Cor.~30]{vAGGdW:quantumSDP} where we include the row sparsity $r$ as a parameter. Their lower bound corresponds to setting $r=d$. Since the bound below is only a minor variation of their bound, we only give a sketch of how to recover it from their arguments (one mainly has to take care to transform their LP from canonical form to the form we use).
\begin{theorem}[{Variation of \cite[Cor.~30]{vAGGdW:quantumSDP}}]
    There is a family of linear programs of the form $\min\ c^T x$ s.t.~$Ax \geq b$ with $A \in \R^{n \times d}$ for which determining the optimal value up to additive error $<1/2$, with probability $\geq 2/3$, requires at least $\Omega(\sqrt{nd}r)$ quantum queries to $A$, or at least $\Omega(nr)$ classical queries to~$A$. 
\end{theorem}

\begin{proof}[Proof sketch]

In~\cite{vAGGdW:quantumSDP} the authors obtain a quantum query lower bound on the cost of solving linear programs by constructing a linear program whose optimal value computes a Boolean function. 

Concretely, consider the $\MAJ_a$-$\OR_b$-$\MAJ_c$ function that takes an input $Z \in \{0,1\}^{a \times b \times c}$ and outputs 
\begin{align*}
\MAJ_a( & \\
        &\OR_b(\MAJ_c(Z_{111},\ldots,Z_{11c}), \ldots, \MAJ_c(Z_{1b1},\ldots,Z_{1bc})), \\
        &\ldots,\\
        &\OR_b(\MAJ_c(Z_{a11},\ldots,Z_{a1c}), \ldots, \MAJ_c(Z_{ab1},\ldots,Z_{abc})) \\
        &)
\end{align*}
The quantum query complexity of $\MAJ_a$-$\OR_b$-$\MAJ_c$ is known to be $\Theta(a \sqrt{b} c)$, which follows from the known quantum query complexities of $\MAJ_a$ and $\OR_b$, and the composition of quantum query complexity~\cite{Reichardt11reflections}. The randomized query complexity of $\MAJ_a$-$\OR_b$-$\MAJ_c$ is similarly known to be $\Theta(a bc)$, using the known randomized query complexities of $\OR_b$ ($\Theta(b)$), $\MAJ_a$ ($\Theta(a)$) and a suitable composition theorem. We note that while the composition of randomized query complexity is an open problem in general, it is known to hold for example when the outer function is symmetric, or has full query complexity (both of which hold here)~\cite{chakraborty23composition}. 

We then consider the following linear program where we let $Z^i$ be the $c \times b$ matrix whose $jk$-th entry is $Z_{ikj}$:
\begin{align*}
\max\ &\ \sum_{k=1}^c w_k \\
\text{s.t. } & \ 
\begin{pmatrix}
I_c & -Z^1 & -Z^2 & \cdots & -Z^a \\ 
& \mathbf{1}_b^T &  & & \\
& & \mathbf 1_b^T & & \\
& & & \ddots & & \\
& & \mathbf 1_b^T & & \mathbf 1_b^T\\
\end{pmatrix} \begin{pmatrix} w \\ v^1 \\ \vdots \\ \\
v^a \end{pmatrix} \leq 
\begin{pmatrix}
0_c \\ \mathbf 1_a     
\end{pmatrix}
\\
&\ w \in \R^c_+, v^1, \ldots, v^a \in \R^b_+
\end{align*}
In \cite{vAGGdW:quantumSDP} it is shown that the optimal value of this linear program is $\sum_{i=1}^a \max_{j \in [b]} \sum_{\ell=1}^c Z_{ij\ell}$; when restricted to the ``hard instances'' for each of the Boolean functions, this integer determines the value of the $\MAJ_a$-$\OR_b$-$\MAJ_c$ function on $Z$ (we refer to \cite{vAGGdW:quantumSDP} for a more detailed description of the conditions on $Z$).

The number of linear inequality constraints is $c+a$ and the number of nonnegative variables is $c+ab$. Each column of the constraint matrix has at most $c+1$ nonzero entries. The primal problem thus has more variables than constraints, we therefore consider the dual linear program, which by strong duality has the same objective value.  The dual of the above LP will thus have $d = c+a$ nonnegative variables, $c+ab$ constraints, and a constraint matrix that has sparsity $r = c+1$. To bring the LP into our desired form we view the $d$ nonnegativity constraints as inequality constraints on $d$ real-valued variables. This brings the total number of linear inequality constraints to $n= 2c + a(b+1)$.
We express $a$, $b$, $c$ in terms of $r$, $n$, $d$ as follows: $c = r-1$, $a = d-r+1$ and  $b=\frac{n-2(r-1)}{d-r+1}-1$. Note that as long as $r \leq d/2$, we have $a = \Theta(d)$, $b = \Theta(n/d)$, and $c = \Theta(r)$. This implies that finding an additive $\eps=1/3$ approximation of the optimal value of the dual LP requires $\Omega(d \sqrt{n/d}r) = \Omega(\sqrt{nd}r)$ quantum queries to the constraint matrix $A$, or $\Omega(nr)$ classical queries to $A$.
\end{proof}

Currently there is a gap between the quantum query lower bound and upper bound. While we expect to be able to improve the upper bound, it is possible that the lower bound can be strengthened.

We finally establish a lower bound that applies to the quantum-inspired sampling and query access model.

\begin{lemma}
    There is a family of linear programs of the form $\min\ c^T x$ s.t.~$Ax \geq b$ with $A \in \R^{n \times d}$ for which determining the optimal value up to additive error $<1/2$, with probability $\geq 2/3$, requires at least $\Omega(n)$ queries to $\mathrm{SQ}(A)$, $\mathrm{SQ}(A^T)$, $\mathrm{SQ}(b)$, and $\mathrm{SQ}(c)$.  
\end{lemma}
\begin{proof}
    Let $c$ be the all-ones vector, $A$ the $n \times d$ matrix that consists of $d$ blocks of size $n/d \times d$: the $j$-th block contains the all-ones vector in the $j$th column and is zero elsewhere. Now consider a bitstring $z \in \{0,1\}^n$ and construct the vector $b \in \{-1,1\}^n$ by setting $b_i = -(-1)^{z_i}$. When $|z|=0$, the optimal value of the LP $\min\ c^T x$ s.t.~$Ax \geq b$ is $-n$, corresponding to the optimal solution whose entries are all $-1$. When $|z|=1$, the optimal value is $-n+2$: the index $i \in [n]$ for which $z_i=1$ forces one coordinate to be at least $1$. Determining whether $|z|=0$ or $|z|=1$ suffices to compute $\OR(z)$, which is known to require $\Omega(n)$ classical queries to $z$. Finally note that the sampling and query access to $A$, $b$, and $c$ can be simulated using at most $1$ query to $z$: $A$ and $c$ are independent of $z$, sampling an index according to $|b_i|^2/\|b\|^2$ amounts to drawing a $j \in [n]$ uniformly at random, and a query to $b$ is equivalent to a query to $z$. 
\end{proof}

\section*{Acknowledgments}

We thank Aaron Sidford for pointing out that Lee's algorithm for approximating Lewis weights~\cite{YinTatThesis} only proves a 1-sided bound.
In a follow-up work \cite{apers2024lewis} we, together with Sidford, describe a new algorithm for obtaining a 2-sided bound, and we use this algorithm here.\footnote{In a first version of this manuscript we claimed Lee's algorithm provides a 2-sided bound. However, as Sidford pointed out to us, it only provides a 1-sided bound. Using the new algorithm from \cite{apers2024lewis} solves this issue, but has increased the degree of the $\poly(d)$ factor in a few places (e.g., solving LPs with the Lewis weight barrier now has quantum query complexity $\tO(\sqrt{n} d^5)$ instead of the previously claimed~$\tO(\sqrt{n} d^3)$).}

We furthermore thank anonymous referees for their careful reading of our manuscript and their useful comments, which helped improve the exposition of our work. 

Part of the exposition of the quantum algorithm for spectral approximation is based on the report of an internship at IRIF by Hugo Abreu and Hugo Thomas, which we are thankful for.
We also thank Adrian Vladu for useful discussions on IPMs and Lewis weights, and Arjan Cornelissen for useful discussions on quantum mean estimation.

This work was supported by QUDATA project ANR-18-CE47-0010 and the European QuantERA project QOPT (ERA-NET Cofund 2022-25). The work of SA was additionally supported in part by the French PEPR integrated projects EPiQ (ANR-22-PETQ0007) and HQI (ANR-22-PNCQ-0002), and the French ANR project QUOPS (ANR-22-CE47-0003-01).

\bibliographystyle{plain}
\bibliography{references}

\appendix 

\section{Quantum algorithm based on cutting plane method}
\label{app: cutting plane}

The currently fastest cutting plane method to solve convex optimization problems is due to Lee, Sidford, and Wong~\cite{LeeSidfordWong15}. For comparison, we state its complexity here and we show how to quantize it in the setting of linear programming. 
\begin{theorem}[Informal, Lee-Sidford-Wong~\cite{LeeSidfordWong15}]
Let $K \subseteq \R^d$ be a convex body contained in a (known) ball of radius $R$ and let $\eps>0$. Then with $O(d \log(dR/\eps))$ queries to a separation oracle for $K$ and $\widetilde O(d^3)$ additional work, we can either find a point in $K$ or prove that $K$ does not contain a ball of radius $\eps$. 
\end{theorem} 

Together with a binary search on the objective value of $\min\  c^T x$ s.t.~$A x \geq b$ and a quantum implementation of a separation oracle (see below), this allows one to solve tall LPs quantumly in time 
\[
\widetilde O(\sqrt{n} d r + d^3).
\]
The above mentioned quantum implementation of a (weak) separation oracle is a simple application of Grover's algorithm. Indeed, given a vector $x \in \R^d$, we need to determine if $Ax \geq b$. If not, then we need to return a hyperplane that separates $x$ from the feasible region. We do so by using Grover's algorithm to search for a violated inequality: a violated inequality will separate $x$ from the feasible region, whereas $x$ is feasible if none are found. Verifying a single inequality $a_i^T x \geq b_i$ for $i \in [n]$ can be implemented using one row query to $A$, a single query to $b$, and time~$\tO(r)$. This implements a separation oracle in time $O(\sqrt{n} r)$.

\section{Omitted proofs about robustness of IPMs}

\label{sec: proofs robust update}

In this section we prove \cref{thm:IPM-master}. As discussed in \cref{sec:robustness}, the arguments are relatively standard. Here we give an exposition that closely follows that of Renegar~\cite{Renegar88,Renegar01}, but similar statements can be found e.g.~in~\cite{NesterovNemirovskii93}.

\subsection{Self-concordant barriers and path-following}

We consider an optimization problem of the form
\[
\min_{x \in D} c^T x
\]
for some open, bounded, full-dimensional, convex region $D \subset \R^d$ and a cost function $c \in \R^d$.
Let $\mathrm{val}$ denote the optimum.
We turn this into an unconstrained optimization problem by using a \emph{barrier function} $f:D \to \R$ that satisfies $f(x) \to \infty$ as $x$ approaches the boundary of $D$.
We assume that $f$ is twice continuously differentiable with gradient $g(x)$ and positive definite Hessian~$H(x)$.

We first recall that the Hessian defines a ``local'' inner product by
\[
\langle a,b \rangle_{H(x)}
\coloneqq \langle a,b \rangle_x
\coloneqq a^T H(x) v
\]
for $a,b \in \R^d$.
In turn this defines the local norm $\|v\|_{H(x)} \coloneqq \|v\|_x \coloneqq \sqrt{v^T H(x) v}$. 
Note that the gradient $g(x)$ and Hessian $H(x)$ are defined with respect to the standard inner product.
We will use $g_x(y):=H(x)^{-1}g(y)$ and $H_x(y) = H(x)^{-1}H(y)$ to denote the gradient and Hessian at $y$ with respect to the local inner product at $x$.
In particular, note that $g_x(x) = H^{-1}(x) g(x) =: - n(x)$ is the Newton step at $x$ and $H_x(x) = I$.

\subsubsection{Self-concordance} \label{sec:SCB}

As an additional requirement, we impose that $f$ is ``self-concordant''.
In the following, let $B_x(y,r)$ denote the open ball of radius $r$ centered at $y$ in the norm $\|\cdot \|_x$. 

\begin{definition}[Self-concordance] \label{def:self conc}
The function $f$ is \emph{self-concordant}\footnote{Technically speaking, we assume strongly nondegenerate self-concordant, where the term ``strongly'' refers to the requirement $B_x(x,1) \subseteq D$, and ``nondegenerate'' refers to a positive definite Hessian at each $x$ (and hence a local inner product at each $x$). } if for all $x \in D$ we have $B_x(x,1) \subseteq D$, and if whenever $y \in B_x(x,1)$ we have 
\[
1-\|y-x\|_x \leq \frac{\|v\|_y}{\|v\|_x} \leq \frac{1}{1-\|y-x\|_x} \quad \text{for all } v \neq 0.
\]
\end{definition}
For example, if in the above \cref{def:self conc} we have $\|y-x\|_x \leq 1/2$, then $\frac{1}{2} \|v\|_x \leq \|v\|_y \leq 2 \|v\|_x$ for all $v$. 

The self-concordance condition is roughly the same as Lipschitzness of the Hessian of $f$, in the following way.  
(Note that $\frac{1}{(1-\|y-x\|_x)^2}-1 \approx 2\|y-x\|_x$ for small $\|y-x\|_x$, and $H_x(x)=I$.)
\begin{theorem}[{\cite[Thrm.~2.2.1]{Renegar01}}] \label{thrm: Hessian bounds for self-concordant}
Assume $f$ has the property that $B_x(x,1) \subseteq D$ for all $x \in D$. Then $f$ is self-concordant if and only if for all $x \in D$ and $y \in B_x(x,1)$ we have 
\begin{equation}
\|H_x(y)\|_x, \|H_x(y)^{-1}\|_x \leq \frac{1}{(1-\|y-x\|_x)^2};
\end{equation}
likewise, $f$ is self-concordant if and only if 
\begin{equation}
\|I-H_x(y)\|_x,\|I-H_x(y)^{-1}\|_x \leq \frac{1}{(1-\|y-x\|_x)^2}-1.
\end{equation}
\end{theorem}

A last quantity that we associate to a barrier function is its \emph{complexity}, which is defined below.
We use it to bound the complexity of path-following in the next section.

\begin{definition}[Barrier complexity]
A functional $f$ is called a \emph{(strongly nondegenerate self-concordant) barrier} if $f$ is self-concordant and 
\[
\vartheta_f:= \sup_{x \in D} \|g_x(x)\|_x^2 < \infty. 
\] 
The value $\vartheta_f$ is called the \emph{complexity} of $f$.
\end{definition}

\subsubsection{Path-following}

Path-following requires us to optimize a self-concordant function, given an approximate minimizer.
For this we can use Newton's method, which makes updates of the form
\begin{equation}
x_+ \coloneqq x- H(x)^{-1} g(x) = x + n(x),
\end{equation}
where $n(x)$ denotes the Newton step at $x$.
The next theorem shows the two key properties of the (exact) Newton step: its size decreases essentially quadratically, and its size bounds the distance to the minimizer of $f$.
Our later analysis for approximate Newton steps builds on this analysis for exact Newton steps.

\begin{theorem}[{(Exact) Newton steps \cite[Thrm.~2.2.4-5]{Renegar01}}] \label{thm:exact-Newton}
Assume $f$ is self-concordant.
\begin{itemize}
\item
If $\|n(x)\|_x <1$, then $\|n(x_+)\|_{x_+} \leq \left(\frac{\|n(x)\|_x}{1-\|n(x)\|_x}\right)^2$.
\item
If $\|n(x)\|_x \leq 1/4$ for some $x \in D$, then $f$ has a minimizer $z$ and 
\[
\|x_+-z\|_x \leq \frac{3\|n(x)\|_x^2}{(1-\|n(x)\|_x)^3}.
\]
Thus
\[
\|x-z\|_x \leq \|n(x)\|_x + \frac{3\|n(x)\|_x^2}{(1-\|n(x)\|_x)^3} \leq (5/6)^2.
\]
\end{itemize}
\end{theorem}

In an IPM we use $f$ to denote the barrier function that encodes the convex set $D$, and we solve the unconstrained problem $\min_x f_\eta(x)$ where
\[
f_\eta(x)
\coloneqq \eta c^T x + f(x),
\]
for $\eta > 0$.
We let $z(\eta)$ denote the minimizer of $f_\eta$, and call $\{z(\eta)\}_{\eta>0}$ the \emph{central path}.
If $f$ is self-concordant with complexity $\vartheta_f$, then we have the following two properties, see Eqs.~(2.12) and~(2.13) in~\cite{Renegar01}.
\begin{enumerate}
\item
As $\eta\to \infty$, $z(\eta)$ approaches the minimizer of $\langle c,x\rangle$ over $D$:
\begin{equation*}
c^T z(\eta)
\leq \mathrm{val} + \frac{1}{\eta} \vartheta_f. 
\end{equation*}
\item
Moreover, for $y \in B_{z(\eta)}(z(\eta),1)$ we have
\begin{equation*}
c^T y
\leq \mathrm{val} + \frac{1}{\eta} \vartheta_f + \frac{1}{\eta} \vartheta_f \|y-z(\eta)\|_{z(\eta)}
\leq \mathrm{val} + \frac{2}{\eta} \vartheta_f. 
\end{equation*}
\end{enumerate}

The idea of \emph{path-following} is then to start from a point close to the central path for some $\eta$ and alternate the following two steps: (i) increase $\eta \rightarrow \eta'$, and 
(ii) move the current iterate $x$ close to $z(\eta')$. At a high level, if we increase $\eta$ by a small enough amount, then the current iterate $x$ is already close enough to $z(\eta')$ to ensure that Newton's method brings us very close to $z(\eta')$ with few iterations. Once $\eta$ is roughly $\vartheta_f/\eps$, any point at a constant distance from $z(\eta)$ is $\eps$-close to optimal; one can show that updating $\eta$ multiplicatively by a factor roughly $1+1/\sqrt{\vartheta_f}$ is a safe choice and hence roughly $\tO(\sqrt{\vartheta_f} \log(\vartheta_f/\eps))$ iterations suffice.  

Traditionally, to ensure that the iterate is close enough to $z(\eta)$, one upper bounds the size of the Newton step $n(x) = -H(x)^{-1}g(x)$ in the local norm by a small constant (see \cref{thm:exact-Newton}). In order to obtain a quantum speedup, we avoid computing the Hessian exactly, instead we work with a spectral approximation $Q(x) \preceq H(x) \preceq C Q(x)$ for some $C \in \tO(1)$. Based on $Q(x)$ one can define $d(x) = Q(x)^{-1} g(x)$ and one of its key properties is the inequality 
\[
\|n(x)\|_x \leq \|d(x)\|_{Q(x)}.
\]
We can thus ensure that the exact Newton step is small by controlling $\|d(x)\|_{Q(x)}$. In \cref{def:approx-Newton-step} we give a precise definition of the type of approximate Newton step that we work with (we also allow an approximation of $g(x)$). The purpose of the next section is to analyze \cref{alg:approxshortstep}, an algorithm that precisely implements one path-following iteration: starting from an $x \in D$ and $\eta$ for which the approximate Newton step is small, it updates $
\eta$ and finds a new $x$ that has a small approximate Newton step with respect to the new $\eta$.

\subsection{Approximate Newton's method}\label{sec: robust update}

Here we analyze an IPM based on \emph{approximate} Newton steps, which we formalize as follows.

\begin{definition}[Approximate Newton step] \label{def:approx-Newton-step}
Let $f$ be a self-concordant barrier with gradient $g(x)$ and Hessian $H(x)$ and suppose there exists a function $Q(x)$ and a real parameter $C \geq 1$ such that 
\begin{equation} \label{eq: Q vs H2}
Q(x) \preceq H(x) \preceq C Q(x).
\end{equation}
Now define
\[
d(x) = - Q(x)^{-1} g(x),
\]
and let $\tilde d(x)$ be a $\zeta$-approximation of $d(x)$ in the local norm, i.e., $\|\tilde d(x) - d(x) \|_x \leq \zeta$.
We say that $\tilde d(x)$ is a $(\zeta,C)$-\emph{approximate Newton step} for $f$ at $x$.
\end{definition}

In our algorithms, $\zeta$ and $C$ do not change between iterations and hence we often drop the prefix~$(\zeta,C)$.
Note that we can obtain an approximate Newton step by approximating the gradient in the \emph{inverse-local} norm.
Indeed, if $\tilde g \in \R^d$ is such that 
\begin{equation} \label{eq: g approx}
\|\tilde g - g(x)\|_{H(x)^{-1}} \leq \zeta,
\end{equation}
then for $\tilde d(x) = - Q(x)^{-1} \tilde g(x)$ we have $\|\tilde d(x) - d(x)\|_x \leq \zeta$.

The following algorithm implements a single iteration of path-following (going from $\eta$ to $\eta'$) using approximate Newton steps.
Recall that closeness to the central path can be ensured by having a small (approximate) Newton step. The algorithm applies to a self-concordant function $f$ with complexity $\vartheta_f$, we define $d(x)$ as above and use a subscript $\eta$ to emphasize that the gradient of $f_\eta$ depends on $\eta$ (the Hessian of $f_\eta$ does not depend on $\eta$, so neither does $Q(x)$).

\begin{algorithm}[ht]
  \caption{The approximate short-step iteration}\label{alg:approxshortstep}
  \Parameters{$\alpha,\beta,\delta,\zeta>0$,  such that $(\alpha+\zeta) \beta + (\beta-1) \sqrt{\vartheta_f} \leq 1/4$}
  \Input{$x \in D$ and $\eta>0$ with $(\zeta,C)$-approximate Newton step $\tilde d_{\eta}(x)$ satisfying $\|\tilde d_{\eta}(x)\|_{Q(x)} \leq \alpha<1$}
  \Output{$x' \in D$ with approximate Newton step $\tilde d_{\eta'}(x')$ satisfying 
  $\|\tilde d_{\eta'}(x')\|_{Q(x')} \leq \alpha$}

    \BlankLine
    Set $\eta' = \beta \eta$, $\delta = 1/(4C)$, $x' = x$ and let $\tilde d_{\eta'}(x')$ be an approximate Newton step from $x'$\;
    \While{$\|\tilde d_{\eta'}(x) \|_{Q(x)} \geq \alpha$}{
    Let $x' \leftarrow x' + \delta \tilde d_{\eta'}(x')$ and let $\tilde d_{\eta'}(x')$ be an approximate Newton step from $x'$;}
    \Return{$x'$}
\end{algorithm}

The correctness and complexity of the algorithm are established in the following theorem.
\begin{theorem} \label{thm:approxshortstep}
\cref{alg:approxshortstep} correctly returns $x' \in D$ with approximate Newton step $\tilde d_{\eta'}(x')$ satisfying $\|\tilde d_{\eta'}(x')\|_{Q(x')} \leq \alpha$.
The algorithm terminates after $O(C^2)$ iterations, and hence requires $O(C^2)$ many $(\zeta,C)$-approximate Newton steps.
\end{theorem}

Combined with the discussion on path-following, this theorem implies \cref{thm:IPM-master}.

We first sketch the structure of the proof \cref{thm:approxshortstep} and give two useful lemmas. 
To analyze the algorithm, we show that after updating $\eta$, the while-loop terminates in a roughly constant number of iterations. To do so, we establish two properties: (1) each iteration decreases the function value by at least a certain amount, and (2) at the start of the while-loop, the function value is close to its minimum.

\begin{lemma} \label{lem: progress bound approx}
Assume $f$ is self-concordant, $x \in D$, and $Q(x)$ is such that $Q(x) \preceq H(x) \preceq C Q(x)$ for $C \geq 1$. Let $d(x) =- Q(x)^{-1} g(x)$ and assume $\tilde d(x)$ is such that $\|\tilde d(x) - d(x)\|_x \leq \zeta$. Let $\delta=1/(4C)$. If $\delta \|\tilde d(x)\|_x \leq 1/2$ and $\zeta \leq \frac{1}{4C}\|\tilde d(x)\|_x$, then we have 
\[
f(x+\delta \tilde d(x)) \leq f(x) - \frac{1}{16C^2} \|\tilde d(x)\|_x^2. 
\]
\end{lemma}
\begin{proof}[Proof of \cref{lem: progress bound approx}]
For $0<\delta<1$ there exists some $\bar \delta \in [0,\delta]$ such that 
\begin{align*}
f(x+\delta \tilde d(x)) &= f(x) + \delta \langle g(x), \tilde d(x)\rangle + \frac{1}{2} \delta^2 \langle \tilde d(x), H(x+\bar \delta \tilde d(x)) \tilde d(x)\rangle \\
&= f(x) + \delta \langle g(x), d(x) \rangle +\delta \langle g(x), \tilde d(x) - d(x)\rangle + \frac{1}{2} \delta^2 \langle \tilde d(x), H(x+\bar \delta \tilde d(x)) \tilde d(x)\rangle
\end{align*}
We upper bound the last three terms on the right hand side. For the first term, since $g(x) = -Q(x) d(x)$, we have $\langle g(x),d(x)\rangle = -\|d(x)\|_{Q(x)}^2 \leq -\|d(x)\|_{x}^2/C$ where in the last inequality we use $Q(x) \succeq \frac{1}{C}H(x)$. For the second term we have 
\begin{align*}
\langle g(x),\tilde d(x)-d(x)\rangle &= \langle d(x),\tilde d(x)-d(x)\rangle_{Q(x)} \\
&\leq \|d(x)\|_{Q(x)} \|\tilde d(x) - d(x)\|_{Q(x)} \\
&\leq  \|d(x)\|_{x} \|\tilde d(x) - d(x)\|_{x},
\end{align*}
where in the last inequality we use that $Q(x) \preceq H(x)$. 
Finally, for the third term, we have 
\begin{align*}
\langle \tilde d(x), H(x+\bar \delta  \tilde d(x)) \tilde d(x)\rangle &= \langle \tilde d(x), H_x(x+\bar \delta \tilde d(x)) \tilde d(x)\rangle_x \\
&= \|\tilde d(x)\|_x^2 + \langle\tilde  d(x), (H_x(x+\bar \delta \tilde d(x)) - I) \tilde d(x)\rangle_x \\
&\leq \|\tilde d(x)\|_x^2 + \|H_x(x+\bar \delta \tilde d(x)) - I\|_x \|\tilde d(x)\|_x^2 \\
&\leq {\|\tilde d(x)\|_x^2}/{(1-\bar \delta\|\tilde d(x)\|_x)^2} \\
&\leq {\|\tilde d(x)\|_x^2}/{(1-\delta\|\tilde d(x)\|_x)^2},
\end{align*}
where in the second inequality we used \cref{thrm: Hessian bounds for self-concordant}. Combining the three estimates leads to the following upper bound: 
\begin{align*}
f(x+\delta \tilde d(x)) &\leq f(x) - \|\tilde d(x)\|_x^2 \left(\frac{\delta}{C}  - \frac{\delta^2}{2(1-\delta\|\tilde d(x)\|_x)^2}\right) + \delta \|\tilde d(x)\|_x \|\tilde d(x) - d(x)\|_x \\
&\leq f(x) - \|\tilde d(x)\|_x^2 \left(\frac{\delta}{C}  - \frac{\delta^2}{2(1-\delta\|\tilde d(x)\|_x)^2}\right) + \frac{\delta}{4C} \|\tilde d(x)\|_x^2. 
\end{align*}
Finally, assume $\delta \|d(x)\|_x \leq 1/2$ for $\delta = \frac{1}{4C}$, then we have 
\[
\frac{\delta}{C}  - \frac{\delta^2}{2(1-\delta\|d(x)\|_x)^2} - \frac{\delta}{4C} \geq \frac{\delta}{C}  - \frac{\delta^2}{2(1/2)^2} - \frac{\delta}{4C} \geq \frac{\delta}{C} - \frac{\delta}{2C} - \frac{\delta}{4C}= \frac{\delta}{4C} = \frac{1}{16C^2}. 
\]
\end{proof}

In each iteration that $\|\tilde d(x)\|_{Q(x)} > \alpha$, we have $\|\tilde d(x)\|_x > \alpha$ and thus \cref{lem: progress bound approx} shows that 
\[
f(x+\delta \tilde d(x)) \leq f(x) - \frac{\alpha^2}{16C^2}
\]
for a suitable choice of $\delta$. It remains to show that initially, after updating $\eta$ to $\eta' = \beta \eta$, $f_{\eta'}(x)$ is close to the minimum of $f_{\eta'}$. The following lemma is a standard bound on the optimality gap in terms of the size of the Newton step. 

\begin{lemma} \label{lem: potential gap}
Assume $f$ is self-concordant and let $x \in D$. If $\|n(x)\|_x \leq 1/4$, then $f$ has a minimizer $z$ and 
\[
f(x) - f(z) \leq \|n(x)\|_x^2 + 3\left(\frac{\|n(x)\|_x}{1-\|n(x)\|_x}\right)^3. 
\]
\end{lemma}
\begin{proof}
Since $f$ is convex, we have for all $x,y \in D$ that 
\[
f(x) - f(y) \leq \langle g(x), x-y\rangle = \langle n(x), x-y\rangle_x.
\]
We can further upper bound the right hand side using the Cauchy-Schwarz inequality: 
\[
\langle n(x), x-y\rangle_x \leq \|n(x)\|_x \|x-y\|_x. 
\]
We apply the above with $y=z$ and it thus remains to upper bound $\|x-z\|_x$. For this we use the second bullet in \cref{thm:exact-Newton}: we have $\|x-z\|_x \leq \|n(x)\|_x + \frac{3\|n(x)\|_x^2}{(1-\|n(x)\|_x)^3}$. Combining the two upper bounds on $f(x)-f(z)$ concludes the proof. 
\end{proof}

We are now ready to prove \cref{thm:approxshortstep}. 

\begin{proof}[Proof of \cref{thm:approxshortstep}]
From \cref{lem: progress bound approx} we obtain that in each iteration of the while loop, the function value of $f$ decreases by at least $(\frac{\alpha}{4C})^2$. To upper bound the number of such iterations, \cref{lem: potential gap} shows that it suffices to prove that at the start of the while-loop $\|n_{\eta'}(x)\|_x$ is small. To do so, we first recall the following useful inequalities. If $Q(x) \preceq H(x) \preceq C Q(x)$, then for any $x \in D$ and $v \in \R^d$ we have 
\[
\|v\|_{Q(x)} \leq \|v\|_x \leq \sqrt{C} \|v\|_{Q(x)}
\quad \text{ and } \quad
\|n(x)\|_x \leq \|d(x)\|_{Q(x)}.
\]
Then, as a first step, note that if $\|\tilde d_\eta(x)\|_{Q(x)} \leq \alpha$ and $\|\tilde d_\eta(x)-d_\eta(x)\|_{x} \leq \zeta$, then $\|d_\eta(x)\|_{Q(x)} \leq \alpha+\zeta$ and therefore $\|n_\eta(x)\|_{x} \leq \alpha+\zeta$. The second step is to bound the size of the Newton step after updating $\eta$. For this we use the following relation:
\[
n_{\eta'}(x) = \frac{\eta'}{\eta} n_{\eta}(x) + \left(\frac{\eta'}{\eta} -1\right) g_x(x). 
\]
Using this relation, the above bound on $n_\eta(x)$, and the triangle inequality, this gives us \[
\|n_{\eta'}(x)\|_{x} \leq \beta \|n_{\eta}(x)\|_x + (\beta-1)\|g_x(x)\|_x \leq (\alpha+\zeta) \beta + (\beta-1)\sqrt{\vartheta_f} \eqqcolon \gamma.
\] 
Recall that by assumption $\gamma \leq 1/4$ and therefore \cref{lem: potential gap} applied to $f_{\eta'}$ shows that 
\[
f_{\eta'}(x) - \min_{z \in D} f_{\eta'}(z) \leq \gamma^2 + 3 \left(\frac{\gamma}{1-\gamma}\right)^3 \leq 7/64.
\]
Choosing $\alpha, \beta, \delta,\zeta>0$ appropriately shows that the while-loop ends after at most $O(C^2)$ iterations. Suitable choices of parameters are for example $\alpha=1/32$, $\zeta = 1/32$, and $\beta = 1 + \frac{1}{8\sqrt{\vartheta_f}} \leq 2$ so that $\gamma = (\alpha+\zeta)\beta + (\beta-1)\sqrt{\vartheta_f} \leq 1/4$, and $\delta= \frac{1}{4C}$ so that also $\delta \|\tilde d_{\eta'}(x)\|_x \leq \delta (\|d_{\eta'}(x)\|_x + \zeta) \leq \delta (\sqrt{C} \|n_{\eta'}(x)\|_x + \zeta) \leq \delta (\sqrt{C} \gamma + \zeta) < 1/2$.
\end{proof}

\begin{remark}[Maintaining feasibility] A key property of self-concordant barriers is that at each interior point, the unit ball in the local norm is contained in the domain. This ensures that we maintain feasibility of the iterates since our steps $\delta \tilde d(x)$ satisfy $\delta \|\tilde d(x)\|_x <1$ (indeed, $\delta<1$ and the approximate Newton step has size $<1/2$, even after updating $\eta$). A second crucial assumption that we have made is that the domain is full-dimensional; this allows us to avoid (potentially costly) projections onto subspaces. 
\end{remark}

\end{document}